%% file: journal_ieee_peerreview.tex
\documentclass[10pt,letterpaper,journal,twocolumn]{IEEEtran}
\newtheorem{corollary}{Corollary}
\newtheorem{definition}{Definition}

\newtheorem{theorem}{Theorem}
\newtheorem{lemma}{Lemma}

\usepackage{enumerate}
\usepackage{amssymb}
\usepackage{amsmath}
\usepackage{verbatim}
\usepackage{enumerate}
\usepackage{framed,xspace}
\usepackage{epsfig}
\usepackage{amssymb, amsmath, latexsym}
\usepackage{setspace}
\usepackage{bbm}
\usepackage{dsfont}
\usepackage{graphicx}
\usepackage{ifthen}
\newenvironment{lemma2}[2][Lemma]{\begin{trivlist} \item[\hskip \labelsep { #1}\hskip \labelsep {#2.}]}{\end{trivlist}}

\newtheorem{lemmaC}{Lemma}[section]

\newcommand{\mX}{\mathcal X}
\newcommand{\mY}{\mathcal Y}
\newcommand{\mF}{\mathcal F}
\newcommand{\mA}{\mathcal A}
\newcommand{\mB}{\mathcal B}

\newcommand{\mS}{\mathcal S}
\newcommand{\mZ}{\mathcal Z}

\newcommand{\mE}{\mathcal E}

\newcommand{\mT}{\mathcal T}

\newcommand{\abs}[1]{\ensuremath{|#1|}}
\newcommand{\C}[2]{{#1}\wedge{#2}}
\newcommand{\K}[2]{{#1} \searrow {#2}}

\newcommand{\cM}{\mathcal{M}}
\def\01{\{0,1\}}
\newcommand{\Markov}[3]{{#1} \leftrightarrow {#2} 
                             \leftrightarrow {#3}}

\newcommand{\OT}[3]{{\binom{#2}{#1}}{\textsf{-OT}^{#3}}}

\newcommand{\RabinOT}[2]{{({#1})\textsf{-RabinOT}^{#2}}}
\newcommand{\OLFE}[1]{{(#1)}{\textsf{-OLFE}}}

\newcommand{\eps}{\varepsilon}
\newcommand*{\sbin}{\{0,1\}}

\DeclareMathOperator{\dis}{D}
\DeclareMathOperator{\Supp}{supp}

\newcommand{\IP}[1]{{\textsf{IP}_{#1}}}

\newcommand{\EQ}[1]{{\textsf{EQ}_{#1}}}
\newcommand{\IS}[2]{{\Iop(#1;#2)}}
\newcommand{\cHS}[2]{{\Hop(#1|#2)}}
\newcommand{\cIS}[3]{{\Iop(#1;#2|#3)}}

\newcommand{\Iop}{I}
\newcommand{\Hop}{H}
\newcommand{\Hshannon}[2]{{\Hop(#1|#2)}}
\newcommand{\CIshannon}[3]{{\Iop(#1;#2|#3)}}
\newcommand{\Ishannon}[2]{{\Iop(#1;#2)}}
\DeclareMathOperator{\maj}{maj}
\DeclareMathOperator{\emin}{min}
\DeclareMathOperator{\emax}{max}

\newcommand{\rmax}{r_{\emax}}
\newcommand{\rmaxE}[3]{r_{\emax}^{#1}(#2|#3)}
\newcommand{\Hmin}{\Hop_{\emin}}

\newcommand{\chmin}[3]{\Hop_{\emin}^{#1}(#2\mid#3)}

\newcommand{\Hmax}{\Hop_{\emax}}

\newcommand{\rhob}{\ensuremath{\bar{\rho}}}
\newcommand{\support}{\textnormal{supp}}
\newcommand{\rminE}[3]{r_{\emin}^#1(#2|#3)}
\newcommand{\rminTE}[3]{ r_{\emin}^{#1}(#2|#3)}
\newcommand{\hh}[4]{\ensuremath{H_{#1}^{#2}({#3})_{#4}}}
\newcommand{\chh}[5]{\ensuremath{H_{#1}^{#2}({#3}|{#4})_{#5}}}
\newcommand{\Chmaxeps}[3]{\Hop_{\emax}^{#1}(#2|#3)}
\newcommand{\Chminteps}[3]{\Hop_{\emin}^{#1}(#2|#3)}
\newcommand{\hminteps}[2]{ \Hop_{\emin}^{#1}(#2)}

\newcommand{\hvn}[2]{\hh{}{}{#1}{#2}} % von Neumann entropy
\newcommand{\chvn}[3]{\chh{}{}{#1}{#2}{#3}}
\newcommand{\proj}[1]{|#1\rangle\langle#1|}

\newcommand{\hmax}[1]{\Hop_{\emax}(#1)}
\newcommand{\hmin}[2]{\Hop_{\emin}(#1)}

\newcommand{\tr}{\operatorname{tr}}

\newcommand{\supp}[1]{\textnormal{supp}\, ( #1 )}
\newcommand{\id}{{\mathbbm{1}}}

\newcommand\bit{\{0,1\}}
\newcommand{\bits}[1]{\bit^{#1}}
\newcommand{\bra}[1]{{\langle {#1}|}}
\newcommand{\ket}[1]{{|{#1}\rangle}}

\newcommand{\ketbra}[2]{{\ket{#1}{\bra{#2}}}}
\newcommand{\vecstate}[1]{\ket{#1}\bra{#1}}

\newcommand{\h}{\ensuremath{\mathcal{H}}}
\newcommand{\hi}[1]{\ensuremath{\mathcal{H}_{\textnormal{#1}}}}
\newcommand{\hA}{\hi{A}}
\newcommand{\hB}{\hi{B}}

\newcommand{\hR}{\hi{R}}

\newcommand{\idx}[2]{{#1}_{#2}}

\newcommand{\rhoA}{\ensuremath{\idx{\rho}{A}}}
\newcommand{\rhoB}{\ensuremath{\idx{\rho}{B}}}

\newcommand{\rhoABZ}{\ensuremath{\idx{\rho}{ABZ}}}

\newcommand{\rhoABC}{\ensuremath{\idx{\rho}{ABC}}}

\newcommand{\psiAB}{\ensuremath{\idx{\psi}{AB}}}

\newcommand{\phiAB}{\ensuremath{\idx{\phi}{AB}}}
\newcommand{\kron}{\otimes}
\newcommand{\ptrace}[2]{\ensuremath{\ptr{#1} (#2)}}
\newcommand{\idi}[1]{\ensuremath{\mathds{1}_{#1}}}
\newcommand{\idA}{\idi{A}}
\newcommand{\idB}{\idi{B}}
\newcommand{\normstates}[1]{\ensuremath{\mathcal{S}_{=}(#1)}}
\newcommand{\ptr}[1]{\operatorname{tr}_{#1}}

\newcounter{protoCount}
\newcounter{protoList}
\newsavebox{\tmpbox}
\newlength{\protobox}

\newcommand{\cancel}[1]{}

\newcommand{\severin}[1]{{\bf[severin: #1]}}

\title{On the Efficiency of Classical and Quantum Secure Function Evaluation}

\begin{document}
%\author{Severin Winkler \and  J\"urg Wullschleger}

\author{Severin Winkler and J\"urg Wullschleger
%\affiliation{Institute of Theoretical Computer Science, ETH Zurich, 8092  Zurich, Switzerland.}
\thanks{S. Winkler is with the Computer Science Department, ETH Z\"urich, Z\"urich, Switzerland (e-mail: severin@severinwinkler.ch).}
\thanks{J. Wullschleger is with Universit\'e de Montr\'eal and McGill University,  Montr\'eal (Qu\'ebec), Canada (e-mail: juerg@wulli.com).}
\thanks{S. Winkler was supported by the Swiss National Science Foundation and an ETHIIRA grant of ETH's research commission. J. Wullschleger was supported by the U.K. EPSRC grant EP/E04297X/1 and the Canada-France NSERC-ANR project FREQUENCY.}
}

%\institute{}
\maketitle
\begin{abstract}
We provide  bounds on the efficiency of secure one-sided output two-party computation of  arbitrary finite functions from trusted distributed randomness in the statistical case. From these results we derive bounds on the efficiency of protocols  that use different variants of OT as a black-box. When applied to  implementations of OT, these bounds generalize most known results to the  statistical case. Our results hold in particular for transformations between a finite number of primitives and for any error.
In the second part we study the efficiency of quantum protocols implementing OT. While most classical lower bounds for perfectly secure reductions of OT to distributed randomness still hold in the quantum setting, we present a statistically secure protocol that violates these bounds by an arbitrarily large factor. We then prove a weaker lower bound that does hold in the statistical quantum setting and implies that even quantum protocols cannot extend OT. Finally, we present two lower bounds for reductions of OT to commitments and a protocol based on string commitments that is optimal with respect to both of these  bounds.
\end{abstract}
%\footnote{A Preliminary version of this work appeared in \cite{WW10}.}

\section{Introduction}
Secure multi-party computation allows two or more distrustful players to jointly compute a function of their inputs in a secure way \cite{Yao82}. Security here means that the players compute the value of the function correctly without learning more than what they can derive from their own input and output.

A primitive of central importance in secure multi-party computation is \emph{oblivious transfer} (OT). In particular, OT is sufficient to execute any two-party computation securely \cite{GolVai87,Kilian88} and OT can be precomputed offline, i.e.,  before the actual inputs to the computation are available, and converted into OTs later. The original form of OT ($\RabinOT{\frac12}{1}$) has been introduced by Rabin in \cite{Rabin81}. It allows a sender to send a bit $x$, which the receiver will get with probability $\frac{1}{2}$, while the sender does not learn whether the message has arrived or not. Another variant of OT, called one-out-of-two bit-OT ($\OT{1}{2}{1}$) was defined in \cite{EvGoLe85}. Here, the sender has two input bits $x_0$ and $x_1$. The receiver gives as input a choice bit $c$ and receives $x_c$ without learning $x_{1-c}$. The sender gets no information about the choice bit $c$. Other important variants of OT are $\OT{t}{n}{k}$ where the inputs are strings of $k$ bits and the receiver can 
choose $t<n$ out of $n$ secrets and $\RabinOT{p}{k}$ where the inputs are strings of $k$ bits and the erasure probability is $p \in [0,1]$.

If the players have access to noiseless classical or quantum communication only, it is impossible to implement information-theoretically secure OT, i.e. secure against an adversary with unlimited computing power. The primitives $\RabinOT{p}{k}$ and $\OT{1}{2}{1}$ are equally powerful \cite{Crepeau87}, i.e., one can be implemented from the other. Numerous reductions of $\OT{1}{n}{k}$ to $\OT{1}{2}{k'}$ are known \cite{BrCrRo86b,CreSan91,BCS96,DodMic99,WolWul05}. There has also been a lot of interest in reductions of OT to weaker primitives. For example, OT can be realized from noisy channels \cite{CreKil88,CrMoWo04,DFMS04,Wullsc09}, noisy correlations~\cite{WolWul04,NaWi06}, or weak variants of oblivious transfer~ \cite{CreKil88,Cachin98,DaKiSa99,BrCrWo03,DFSS06,Wullsc07}. 

In the quantum setting, OT can be implemented from black-box commitments~\cite{BBCS92,Yao95,DFLSS09,Unruh10}; this reduction is impossible in the classical setting\footnote{The existence of a classical reduction of OT to bit commitment in the malicious model would imply a semi-honest OT protocol from a communication channel only.}.

Given these positive results it is natural to ask how efficient such reductions can be in principle, i.e., how many instances of a given primitive are needed to implement OT.

\subsection{Previous Results}

Several lower bounds for OT reductions are known. The earliest impossibility result for information-theoretically secure reductions of OT~\cite{Beaver96} shows that the number of $\OT{1}{2}{1}$ cannot be \emph{extended}, i.e., there does not exist a protocol using $n$ instances of $\OT{1}{2}{1}$ that perfectly implements $m > n$ instances. Lower bounds on the number of instances of OT needed to perfectly implement other variants of OT have been presented in \cite{DodMic99} (see also \cite{Maurer99}). These bounds have been strengthened and generalized to secure sampling of arbitrary two-party distributions in \cite{WolWul05,WolWul08,PP10,PP11}. These bounds apply to the semi-honest model (where dishonest players follow the protocol, but try to gain additional information from the transcript of the computation) and in the case of implementations of OT also to the malicious model (where dishonest players behave arbitrarily). In the malicious model these bounds can be improved \cite{KKK08}. Lower bounds on the 
number of ANDs needed to implement general functions have been presented in \cite{BeiMal04}.

These results only consider \emph{perfect} protocols and do not give much insight into the case of statistical implementations. As pointed out in \cite{KKK08}, their result \emph{only} applies to the perfect case, because there is a statistically secure protocol that is more efficient \cite{CreSav06}. There can be a large gap between the efficiency of perfect and statistical protocols, as shown in \cite{BeiMal04}: The number of OTs needed to compute the equality function is exponentially bigger in the perfect case than in the statistical case. Therefore, it is not true in general that a bound in the perfect case implies a similar bound in the statistical case.

So far very little is known in the statistical case. In \cite{AhlCsi07} a proof sketch of a lower bound for statistical implementations of $\OT{1}{2}{k}$ has been presented. However, this result only holds in the asymptotic case, where the number $n$ of resource primitives goes to infinity and the error goes to zero as $n$ goes to infinity. In \cite{BeiMal04} a non-asymptotic lower bound on the number of ANDs needed for one-sided secure computation of arbitrary functions with \emph{Boolean} output has been shown. This result directly implies lower bounds for protocols that use $\OT{t}{n}{k}$ as a black-box. However, besides being restricted to Boolean-valued functions this result is not strong enough to show optimality of several known reductions and it does not provide bounds for reductions to randomized primitives such as $\RabinOT{\frac12}{1}$. The impossibility results for perfectly secure implementations of randomized two-party primitives of \cite{PP10,PP11} should also generalize to the case of a small 
statistical error according to the authors.

In the quantum setting almost all known negative results show that a certain primitive is impossible to implement from scratch. Commitment has been shown to be impossible in the quantum setting in \cite{Mayers97,LoChau97}. Using a similar proof, it has been shown in \cite{Lo97} that general one-sided two-party computation and in particular oblivious transfer are also impossible to implement securely in the quantum setting.

The only lower bounds for quantum protocols where the players have access to resource primitives (such as different variants of OT) have been presented in \cite{SaScSo09b}, where Theorem~4.7 shows that important lower bounds for classical protocols also apply to \emph{perfectly} secure quantum reductions.

\subsection{Contributions}
In Section~\ref{sec:sfe-lower-bounds} we consider statistically secure protocols that compute a function between two parties from trusted randomness distributed to the players. We provide two bounds  on the efficiency of such reductions --- in terms of the conditional Shannon entropy and the mutual information of the randomness --- that allow us in particular to derive bounds on the minimal number of $\OT{1}{n}{k}$ or $\RabinOT{p}{k}$ needed to compute a general function securely. Our results hold in the non-asymptotic regime, i.e., we consider a finite number of resource primitives and our results hold for \emph{any} error.

We will use the formalism of smooth entropies to show that one of these two bounds can be generalized to a bound in terms of the conditional min-entropy. This leads to tighter bounds in many cases and to arbitrarily better bounds for some reductions.

In Section \ref{lower-bounds:ot} we provide an additional bound for the special case of statistical implementations of $\OT{1}{n}{k}$ in the semi-honest model. Lower bounds for implementations of OT in the semi-honest model imply similar bounds in the malicious model (cf. Section~\ref{subsec:sec:mal:ot} and Appendix~\ref{appendix:sec:mal}). The bounds for implementations of  $\OT{1}{n}{k}$ (Theorem \ref{thm:impossibility}) imply the following corollary that gives a general bound on the conversion rate between different variants of OT.

\begin{corollary} \label{cor:main2}
For any reduction that implements $M$ instances of $\OT{1}{N}{K}$ from $m$ instances of $\OT{1}{n}{k}$ in the semi-honest model with an error of at most $\eps$, we have
\begin{align*}
\frac{m}{M} \geq \max\left( \frac{(N-1) K}{(n-1) k}, \frac{K}{k},\frac{\log N}{\log n}\right) -  7 N K  \cdot (\eps + h(\eps))\;.
\end{align*}
\end{corollary}

Corollary \ref{cor:main2} generalizes the lower bounds from \cite{DodMic99,WolWul05,WolWul08} to the statistical case and is strictly stronger than the impossibility bounds from \cite{AhlCsi07}. If we let $M=m+1$, $N=n=2$ and $K=k=1$, we obtain a stronger version of  Theorem~3 from \cite{Beaver96} which states that OT cannot be extended. Note that the impossibility results for perfectly secure implementations of randomized two-party primitives of \cite{PP10,PP11} deliver stronger bounds in general (cf. Example 4.1 in \cite{PP11}), and according to the authors these results should also generalize to the case of a small statistical error. However, in contrast to our results they are restricted to randomized primitives only and do not apply to general two-party functions.

Our lower bounds show that the following protocols are (close to) optimal in the sense that they use the minimal number of instances of the given primitive.

\begin{itemize}
 \item The protocol in \cite{BrCrSa96,DodMic99} which uses $\frac{N-1}{n-1}$ instances of $\OT{1}{n}{k}$ to implement $\OT{1}{N}{k}$ is optimal.
 \item The protocol in \cite{WolWul05} which uses $t$ instances of $\OT{1}{n}{k n^{t-1}}$ to implement $\OT{1}{n^t}{k}$ is optimal.
 \item In the semi-honest model, the trivial protocol that implements $\OT{1}{2}{k}$ from $k$ instances of 
 $\OT{1}{2}{1}$ is optimal. In the malicious case, the protocol in \cite{CreSav06} uses asymptotically (as $k$ goes to infinity) the same amount of instances and is therefore asymptotically optimal.
 \item The protocol in \cite{savvides:diss} that implements  $\OT{1}{2}{k}$ from $\RabinOT{\frac{1}{2}}{1}$ in the malicious model is asymptotically optimal.
\end{itemize}

While previous results suggest that quantum protocols are not more efficient than classical protocols for reductions between different variants of oblivious transfer, we present in Section~\ref{sec:quantumReductions} a statistically secure protocol that violates the classical bounds and the bound for perfectly secure quantum protocols by an arbitrarily large factor. More precisely, we prove that, in the quantum setting, string oblivious transfer can be reversed much more efficiently than by any classical protocol.
\noindent
We show that a weaker lower bound for quantum reductions holds also for quantum protocols in the statistical setting (Theorem \ref{thm:imposs:rand}). This result implies in particular that quantum protocols cannot extend oblivious transfer, i.e., there exists a constant $c>0$ such that any quantum reduction of $m+1$ instances of $\OT{1}{2}{1}$ to $m$ instances of $\OT{1}{2}{1}$ must have an error of at least~$\frac{c}{m}$. %Furthermore, Theorem~\ref{thm:imposs:rand} implies a lower bound for reductions between different variants of OT. 
Finally, we also derive a lower bound on the number of commitments (Theorem~\ref{thm:imposs1}) and on the total number of bits the players need to commit to (Theorem~\ref{thm:imposs:com}) in any $\eps$-secure implementation of $\OT{1}{2}{k}$ from commitments.

\begin{corollary}\label{cor:imposs:com}
A protocol that implements $\OT{1}{2}{k}$, using commitments only, with an error of at most $0<\eps\leq 0.002$ must use at least $\log(1/\eps)-6$ individual commitments and needs to commit to at least
$(1- 3 \sqrt{\eps}) \cdot k - 3 h(\sqrt{\eps})$
bits in total. 
\end{corollary}

%\severin{A Preliminary version of this work appeared in \cite{WW10}.}

\section{Preliminaries}
We denote the distribution of a random variable $X$ by $P_X(x)$. Given the distribution $P_{XY}$ over $\mX \times \mY$, the marginal distribution is denoted by $P_{X}(x) := \sum_{y \in \mY} P_{XY}(x,y)$. 
For every $y \in \mY$ with $P_Y(y)>0$, the conditional distribution $P_{X \mid Y}(x,y):=P_{XY}(x,y)/P_Y(y)$ over $\mX \times \mY$ defines a distribution $P_{X \mid Y=y}$ with $P_{X \mid Y=y}(x)=P_{X|Y}(x,y)$ over $\mX$. We say that $X$, $Y$ and $Z$ form a Markov chain, denoted by $\Markov{X}{Y}{Z}$, if $X$
and $Z$ are independent given $Y$, which means that $P_{X|Y =y} = P_{X|Y =y;Z=z}$ 
for all $y,z$. Given an event $\Omega$ and random variables $X$ and $Y$ with a joint distribution $P_{\Omega XY}$, we use the notation $P_{X\Omega|Y=y}$ for the sub-normalized distribution with $P_{X\Omega|Y=y}(x):=P_{X|Y=y}(x)P_{\Omega|X=x,Y=y}(1)\;.$ We will also use the shorthand notation $P_{\Omega|X=x}$ to denote the probability $P_{\Omega|X=x}(1)$. We use the convention that $P_{X\Omega|Y=y}(x)=0$ if $P_Y(y)=0$.

The \emph{statistical distance} between the distributions $P_X$ and $P_{X'}$ over the domain $\cal{X}$ is defined as the maximum, over all (inefficient) distinguishers $\delta:\mX \rightarrow \{0,1\}$, of the distinguishing advantage:
\[\dis(P_X,P_{X'}) :=\max_{\delta}\mid \Pr[\delta(X) =1] - \Pr[\delta(X') = 1]\mid.\]
If $\dis(P_X,P_{X'}) \leq \eps$, we may also say that $P_X$ is $\eps$-close to $P_{X'}$. The support of a distribution $P_X$ over $\mX$ is defined as $\supp{P_X}:=\{x \in \mX:~P_X(x)>0\}$. If $x=(x_1,\dots,x_n)$ and  $T:=\{i_{1},\ldots,i_{k}\}\subseteq \{1,2,\ldots,n\}$, then $x_T$ denotes the sub-string $(x_{i_{1}},x_{i_{2}}, \ldots,x_{i_{k}})$ of $x$. If $x,y \in \bits n$, then $x \oplus y$ denotes the bitwise XOR of $x$ and $y$. 

\subsection{Information Theory}
The \emph{conditional Shannon entropy} of $X$ given $Y$ is defined as\footnote{All logarithms 
are binary.}
%we use the convention that $0 \cdot \log 0 = 0$\;,} 
\[ \cHS{X}{Y} := - \sum_{(x,y) \in \supp{P_{XY}}} P_{XY}(x,y) \log P_{X \mid Y}(x , y) \;.\]
The \emph{mutual information} of $X$ and $Y$ given $Z$ is defined as
\[ \CIshannon{X}{Y}{Z} := \cHS{X}{Z} - \cHS{X}{YZ}\;.\]
Note that $\Markov{X}{Y}{Z}$ if and only if $\CIshannon{X}{Z}{Y} = 0$, i.e., $X$ and $Z$
are conditionally independent given $Y$.
We use the notation
\[ h(p) := -p \log p - (1-p) \log (1-p)\;\]
for the binary entropy function, i.e., $h(p)$ is the Shannon entropy of a binary random variable that takes on one value with probability $p$ and the other with $1-p$. Note that the function $h(p)$ is \emph{concave}, which implies that for any $0 \leq p \leq 1$ and $0 \leq c \leq 1$, we have
\begin{align} \label{eq:h}
 h(c \cdot p) \geq c \cdot h(p)\;. 
\end{align}
We will need the chain-rule
\begin{align} \label{eq:chain}
 \cHS{X Y }{Z} = \cHS{X}{Z}  + \cHS{Y}{X Z}\;,
\end{align}
and the following monotonicity inequalities
\begin{align} \label{eq:mono}
\cHS{X Y}{Z} & \geq \cHS{X}{Z}  \geq \cHS{X}{Y Z}\;, \\
\label{eq:mono-I}\CIshannon{WX}{Y}{Z} & \geq \CIshannon{X}{Y}{Z}\;,\\
\CIshannon{WX}{Y}{Z} & \geq \CIshannon{X}{Y}{ZW}\;.
\end{align}
We will also need
\begin{align} \label{eq:h-average}
  \cHS{X}{Y Z} = \sum_z P_Z(z) \cdot \cHS{X}{Y, Z=z}\;.
\end{align}
$\Markov{X}{Y}{Z}$ implies that
\begin{align} \label{eq:H-markov}
 \cHS{X}{Z} \geq \cHS{X}{YZ} = \cHS{X}{Y}\;.
\end{align}
It is easy to show that if $\Markov{W}{XZ}{Y}$, then
\begin{align} \label{eq:I-markov}
\CIshannon{X}{Y}{Z W} &\leq \CIshannon{X}{Y}{Z}\;\text{and} \\
\label{eq:I-markov2}   \CIshannon{W}{Y}{Z} &\leq \CIshannon{X}{Y}{Z}\;.
\end{align}
%Let $(X,Y)$, and $(\hat X,\hat Y)$ be random variables distributed according to $P_{XY}$ and $P_{\hat X \hat Y}$, and let $\delta(P_{XY},P_{\hat X \hat Y})\leq \epsilon$. 
We will need the following lemma that we prove in Appendix~\ref{app:lemmas}.
\begin{lemma}
\label{lem:h-smooth}
Let $(X,Y)$, and $(\hat X,\hat Y)$ be random variables distributed according to $P_{XY}$ and $P_{\hat X \hat Y}$, and let $\dis(P_{XY},P_{\hat X \hat Y})\leq \epsilon$. Then 
\begin{align*}
\cHS{\hat X}{\hat Y}&\geq\cHS{X}{Y}-\epsilon \log |\mathcal{X}|-h(\epsilon)\;.
%\Ishannon(\hat X;\hat Y)&\geq \Ishannon(X;Y)-\epsilon \log(|\mathcal{X}|)-\hop(\epsilon)\;.
\end{align*}
\end{lemma}
Lemma \ref{lem:h-smooth} implies Fano's inequality: For all $X, \hat X \in \mX$ with $\Pr[X \neq \hat X] \leq \eps$, we have
\begin{align} \label{eq:fano}
  \cHS{X}{\hat X} \leq \eps \cdot \log |\mX| + h(\eps)\;.
\end{align}

\subsection{Smooth Entropies}
The \emph{min-entropy} $\hmin{X}{}$ is the negative logarithm of the probability of the most likely element
%\[\hmin{X}{}=\lim_{\alpha\to \infty}H_{\alpha}(X)=-\log\max_x P_X(x).\]
\[\hmin{X}{}:=-\log\max_x P_X(x)\;.\]
The \emph{max-entropy} is defined as the logarithm of the size of the support of $P_X$
\[\hmax{X}{}:=\log|\supp{P_X}|\;.\]
There is no standard definition of conditional min- or max-entropy. A natural definition of the min-entropy\footnote{This definition has been introduced in \cite{DodisRS04} in the context of cryptography. Furthermore, it corresponds to the definition of quantum conditional min-entropy \cite{Renner2005} for the special case of classical states.}  is the following 
\begin{align*}
 \Chminteps{}{X}{Y}:&=-\log\sum_y P_Y(y)\max_{x} P_{X|Y=y}(x)\\
&=-\log\sum_y \max_{x} P_{XY}(x,y)\;.
\end{align*}
%\Chminteps{}{X}{Y}:=-\log\sum_y P_Y(y)\max_{x} P_{X|Y=y}(x)=-\log\sum_y \max_{x} P_{XY}(x,y).\]
Then $2^{-\Chminteps{}{X}{Y}}$ corresponds to the maximal probability to guess $X$ from $Y$. %An alternative definition has been used in \cite{RenWol05}, where the min-entropy has been defined as \[\Chmineps{}{X}{Y}:=\min_{x,y}(-\log(P_{X|Y=y}(x)).\]
In contrast to Shannon entropy, min- and max-entropies are not robust to small changes in the distribution. Therefore, one often considers smoothed versions of these measures, where the entropy is optimized over a set of distributions that are close in terms of some distance measure. While the concept of smoothed entropies has already been used in the literature on randomness extraction \cite{NZ96}, the term  \emph{smooth entropy} has been introduced in \cite{RenWol05}. There it has been shown that the smoothed conditional min- and max-entropy\footnote{The variant of conditional min-entropy used there is different from the one we consider in this article.} have similar properties as the Shannon entropy, i.e., they satisfy a chain rule, monotonicity and subadditivity. 

%We define the smoothed min- and max-entropies.
\begin{definition}
For random variables $X,Y$ and $\eps \in [0,1)$, we define
\begin{align*}
 \Chmaxeps{\eps}{X}{Y}&:=\min_{\Omega:\Pr[\Omega]\geq 1-\eps}\max_{y}\log|\supp{P_{X\Omega|Y=y}}|\\
 %\chmin{\eps}{X}{Y}&:=\max_{\Omega:\Pr[\Omega]>1-\eps}\min_{x,y}(-\log(P_{X\Omega|Y=y}(x))\\
 \Chminteps{\eps}{X}{Y}&:=\max_{\Omega:\Pr[\Omega]\geq 1-\eps}-\log\sum_{y}P_{Y}(y) \max_{x} P_{X\Omega|Y=y}(x).
\end{align*}
%where $\mY$ is the support of $P_Y$.
\end{definition}

\cancel{
Note that the two notions of smooth conditional min-entropy are equivalent up to an additive term $\log(1/\eps)$ when smoothed, i.e., $ \Chminteps{\eps+\eps'}{X}{Y}\geq\Chminteps{\eps}{X}{Y}-\log(1/\eps')$, which follows from Markov's inequality. 
}
In Appendix~\ref{app:smooth} we prove various properties of the entropies $\Chminteps{\eps}{X}{Y}$ and $\Chmaxeps{\eps}{X}{Y}$.

\subsection{Primitives and Randomized Primitives}
In the following we consider two-party primitives that take inputs $x$ from Alice and $y$ from Bob and output $\bar x$ to Alice and $\bar y$ to Bob, where $(\bar x, \bar y)$ are distributed according to $P_{\bar X \bar Y \mid XY}$. For simplicity, we identify such a primitive with $P_{\bar X \bar Y \mid XY}$. If the primitive has no input and outputs values $(u,v)$ distributed according to $P_{U V}$, we may simply write $P_{U V}$.
If the primitive is deterministic and only Bob gets an output, i.e., if there exists a function $f:\mX \times \mY \rightarrow \mZ$ such that  $P_{\bar X \bar Y \mid X=x,Y=y}(\perp,f(x,y))=1$ for all $x,y$, then we identify the primitive with the function~$f$.
 
Examples of such primitives are $\OT{t}{n}{k}$, $\RabinOT{p}{k}$, $\EQ{n}$ and $\IP{n}$:

\begin{itemize}
 \item 
 $\OT{t}{n}{k}$ is the primitive where Alice has an input $x = (x_0, \dots, x_{n-1}) \in \{0,1\}^{k \cdot n}$, and Bob has an input $c \subseteq \{0, \dots, n-1\}$ with $|c|=t$.  Bob receives $y = x|_c \in \{0,1\}^{tk}$.
\item
 $\RabinOT{p}{k}$ is the primitive where Alice has an input $x \in \{0,1\}^{k}$. Bob receives $y$ which is equal to $x$ with probability $p$ and $\Delta$ otherwise.
\item
The \emph{equality} function $\EQ{n}:\{0,1\}^n\times \{0,1\}^n\rightarrow \{0,1\}$ is defined as
%$\EQ{n}(x,y)=1$ if $x=y$ and $\EQ{n}(x,y)=0$ otherwise.
\begin{align*}
\EQ{n}(x,y):=\begin{cases}
  1, & \text{if } x=y\;,\\
  0, & \text{otherwise }\;.
\end{cases}
\end{align*}
\item The \emph{inner-product-modulo-two} function $\IP{n}:\{0,1\}^n\times \{0,1\}^n \rightarrow \{0,1\}^n$ is defined as \linebreak$\IP{n}(x,y) := \oplus_{i=1}^n x_iy_i$.
\end{itemize}

We often allow a protocol to use a primitive $P_{UV}$ that does not have any input and outputs $u$ and $v$ distributed according to the distribution $P_{UV}$ to the players. This is enough to model reductions to $\OT{t}{n}{k}$ and $\RabinOT{p}{k}$, since these primitives are equivalent to distributed randomness $P_{UV}$, i.e., there exist two protocols that are secure in the semi-honest model: one that generates the distributed randomness using \emph{one} instance of the primitive, and one that implements \emph{one} instance of the primitive using the distributed randomness as input to the two parties. The fact that $\OT{1}{2}{1}$ is equivalent to distributed randomness has been presented in \cite{BBCS92,Beaver95}. The generalization to $\OT{t}{n}{k}$ is straightforward. The randomized primitives are obtained by simply choosing all inputs uniformly at random. For $\RabinOT{p}{k}$, the implementation is straightforward. Hence, any protocol that uses some instances of $\OT{t}{n}{k}$ or $\RabinOT{p}{k}$
%or $\OLFE{q}$
can be converted into a protocol that only uses a primitive $P_{UV}$ without any input. 

\subsection{Protocols and Security in the Semi-Honest Model} \label{sec:prot}
We will consider the following model: The two parties use a primitive $P_{UV}$ that has no input and outputs values $(u,v)$ distributed according to $P_{U V}$ to the players. Alice and Bob receive inputs $x$ and $y$. Then, the players exchange messages in several rounds, where we assume that Alice sends the first message. If $i$ is odd, then Alice computes the $i$-th message as a randomized function of all previous messages, her input $x$ and $u$. If $i$ is even, then Bob computes the $i$-th message as a randomized function of all previous messages, his input $y$ and $v$. We assume that the number of rounds is bounded by a constant $t$. By padding the protocol with empty rounds, we can thus assume without loss of generality that the protocol uses $t$ rounds in every execution. After $t$ rounds, %Alice computes her output $\tilde x$ as a randomized function of $(M,U,x)$ and 
Bob computes his output $\tilde z$ as a randomized function of $(M,V,y)$, where  $M=(M_1,\ldots,M_t)$ is the sequence of all messages exchanged. Let $P_{XYUV}=P_{X}P_YP_{UV}$, i.e., each party chooses her input independently of $UV$ and the other party's input. Thus, it holds that 
\begin{align}
\cIS{X}{Y}{V}=0\;,\label{eq:markov_first}\\
\cIS{X}{YV}{U}=0\;.\label{eq:markov_second}
\end{align}
Since Bob generates his output $\tilde Z$ from $YVM$, it holds that
\begin{align}\label{equ:markov_prot}
\cIS{\tilde Z}{XU}{YVM} = 0 \;.
\end{align}
Let $M_0:=\bot$ and $M^i := (M_0, \dots, M_i)$. From the definition of the protocol and the distribution of the inputs, we can conclude that
\begin{align}
\cIS{M_{2i+1}}{Y}{M^{2i}XV}&=0\;,\label{ineq:alice_round}\\
\cIS{M_{2i+1}}{YV}{M^{2i}XU}&=0\label{ineq:alice_round2}
\end{align}
for all $i\in\{0,\cdots,(t-1)/2\}$, and 
\begin{align}
\cIS{M_{2i}}{XU}{M^{2i-1}YV}=0\label{ineq:bob_round_g}
\end{align}
for all $i\in\{1,\cdots,(t-1)/2\}$. Applying the monotonicity inequalities for conditional mutual information to equality~\eqref{ineq:bob_round_g} yields the two equalities
\begin{align}
\cIS{M_{2i}}{X}{M^{2i-1}YV}&=0\;,\label{ineq:bob_round}\\
\cIS{M_{2i}}{X}{M^{2i-1}YUV}&=0\;.\label{ineq:bob_round2}
\end{align}
By repeatedly applying inequality~\eqref{eq:I-markov} and the Markov chain relations \eqref{ineq:alice_round} and \eqref{ineq:bob_round} we obtain that
\begin{align}
\cIS{Y}{X}{VM}\leq \cIS{Y}{X}{V}\nonumber\;.%\Markov{X}{VM}{Y}.
\end{align}
Equality~\eqref{equ:markov_prot} implies $\cIS{\tilde Z}{X}{VMY}=0$. Thus, we can conclude, together with equality~\eqref{eq:markov_first}, that
\begin{align}\label{ineq:markov_first}
\cIS{\tilde Z Y}{X}{VM}\leq \cIS{Y}{X}{VM}\leq \cIS{Y}{X}{V}=0\;.%$\Markov{\tilde ZY}{YVM}{X}$
\end{align}
Similarly, the Markov chain relations~\eqref{ineq:alice_round2} and \eqref{ineq:bob_round2} imply, together with equality~\eqref{eq:markov_second}, that
\begin{align}\label{ineq:markov_second}
\cIS{\tilde ZYV}{X}{UM}\leq\cIS{YV}{X}{UM}\leq \cIS{YV}{X}{U}=0\;.
\end{align}
%Since $\cIS{X}{Y}{V}=0$ and $\cIS{X}{YV}{U}=0$, inequalities~\eqref{ineq:markov_first} and~\eqref{ineq:markov_second} imply that
%\begin{align}
%\cIS{\tilde ZY}{X}{VM}=\cIS{\tilde ZYV}{X}{UM}=0\;.\label{ineq:prot_markov}
%\end{align}

\cancel{
It is easy to check that inequalities \eqref{eq:I-markov} and \eqref{eq:I-markov2} imply that, for every distribution of the inputs $X$ and $Y$, %we have $\cIS{\tilde X}{YV\tilde Y}{XUM}=0$ and $\cIS{\tilde Y}{XU\tilde X}{YVM}=0$, where  $M=(M_1,\ldots,M_t)$ is the sequence of all messages exchanged.
we have $\cIS{\tilde Z}{XU}{YVM}=0$, $\cIS{\tilde ZY}{X}{VM}=0$ and $\cIS{\tilde ZYV}{X}{UM}=0$.}

We will consider the \emph{semi-honest model}, where both players behave honestly, but may save all the information they get during the protocol to obtain extra information about the other player's input or output. A protocol securely implements $f:\mX \times \mY \rightarrow \mZ$ with an error $\eps$, if the entire view of each player can be simulated with an error of at most $\eps$ in an ideal setting, where the players only have black-box access to the primitive $f:\mX \times \mY \rightarrow \mZ$. Note that this simulation is allowed to change neither the input nor the output. This definition of security follows Definition 7.2.1 from \cite{Goldreich04}, but is adapted to the case of computationally unbounded adversaries and statistical indistinguishability. 

\begin{definition}\label{definition:sec}
Let $\Pi$ be a \emph{protocol with black-box access to a primitive $P_{U V}$} that implements a primitive $f:\mX \times \mY \rightarrow \mZ$. The random variables $View_A^{\Pi}(x,y)$ and $View_B^{\Pi}(x,y)$ denote the views of Alice and Bob on input $(x,y)$ defined as $(x,u,m_1,\dots,m_t,r_A)$ and $(y,v,m_1,\dots,m_t,r_B)$, respectively, where $r_A$  ($r_B$) is the private randomness of Alice (Bob), $m_i$ represents the $i$-th message and $u,v$ is the output from $P_{UV}$. $Out_B^{\Pi}(x,y)$ denotes the output (which is implicit in the view) of Bob on input $(x,y)$. The protocol is secure in the semi-honest model with an error of at most $\eps$, if there exist two randomized functions $S_A$ and $S_B$, called the simulators\footnote{We do not require the simulator to be efficient.}, such that for all $x$ and~$y$:
\begin{align*}
\dis((View_A^{\Pi}(x,y),Out_B^{\Pi}(x,y)),(S_A(x),z))&\leq \eps\;,\\
\dis((View_B^{\Pi}(x,y),Out_B^{\Pi}(x,y)),(S_B(y,z),z))&\leq \eps\;,
\end{align*}
%where $\bar x,\bar y$ are distributed according to $P_{\bar X \bar Y \mid X=x,Y=y}$.
where $z=f(x,y)$.
\end{definition}

Note that security in the semi-honest model does not directly imply security in the malicious model, as the simulator is allowed to change the input/output in the malicious model, while he is not allowed to do so in the semi-honest model. We will, therefore, also consider security in the \emph{weak semi-honest model}, which is implied both by security in the semi-honest model and by security in the malicious model. Here, the simulator is allowed to change the input to the ideal primitive and change the output from the ideal primitive. Thus, in order to show impossibility of certain protocols in the malicious and in the semi-honest model, it is sufficient to show impossibility in the weak semi-honest model.
\cancel{We also consider the \emph{weak semi-honest} model, where the players also follow the protocol honestly, but the simulator is in contrast to the semi-honest model allowed to change the input to the ideal primitive and change the output from the ideal primitive. Note that both security in the semi-honest model and security in the malicious model imply security in the weak semi-honest model. Thus, in order to show impossibility of certain protocols in the malicious and in the semi-honest model, it is sufficient to show impossibility in the weak semi-honest model.}

\subsection{Sufficient Statistics}
Intuitively speaking, the sufficient statistics of $X$ with respect to $Y$, denoted $\K{X}{Y}$, is the part of $X$ that is correlated with $Y$.

\begin{definition}
Let $X$ and $Y$ be random variables, and let $f(x)=P_{Y|X=x}$. The sufficient statistics of $X$ with respect to $Y$ is defined as $ \K{X}{Y}=f(X)$.
\end{definition}
It is easy to show (see for example \cite{FiWoWu04}) that for any $P_{XY}$, we have
$\Markov{X}{\K{X}{Y}}{Y}$.
This immediately implies that any protocol with access to a primitive $P_{UV}$ can be transformed into a protocol with access to $P_{\K{U}{V},\K{V}{U}}$ (without compromising the security) because the players can compute  $P_{UV}$ from $P_{\K{U}{V},\K{V}{U}}$ privately. Thus, in the following we only consider primitives $P_{UV}$ where $U=\K{U}{V}$ and $V=\K{V}{U}$.

\subsection{Common Part}

The common part was first introduced in \cite{GacKoe73}. In
a cryptographic context, it was used in \cite{WolWul04}. Roughly speaking, the common part $\C{X}{Y}$ of $X$ and $Y$ is the maximal
element of the set of all random variables (i.e., the {\em finest\/} random variable) 
that can be generated both
from $X$ and from $Y$ without any error. For example, if $X=(X_0,X_1) \in \{0,1\}^2$ and $Y=(Y_0,Y_1) \in \{0,1\}^2$,
and we have $X_0 = Y_0$ and $\Pr[X_1 \neq Y_1] = \eps > 0$, then the common part of $X$ and $Y$ is
equivalent to $X_0$. 

\begin{definition}
  Let $X$ and $Y$ be random variables with distribution $P_{XY}$. Let $\mX:=\Supp(P_X)$ and $\mY:=\Supp(P_Y)$. Then $\C{X}{Y}$, the \emph{common part} of $X$ and $Y$, is constructed in the following way:

  \begin{itemize}

  \item Consider the bipartite graph $G$ with vertex set $\mX\cup\mY$,
  and where two vertices $x\in\mX$ and $y\in\mY$ are connected by an
  edge if $P_{XY}(x,y) > 0$ holds.

  \item Let $f_X\, :\, \mX \rightarrow 2^{\mX\cup\mY}$ be the function
  that maps a vertex $v \in \mX$ of $G$ to the set of vertices in the
  connected component of $G$ containing $v$. Let $f_Y\, :\,
  \mY\rightarrow 2^{\mX\cup\mY}$ be the function that does the same
  for a vertex $w \in \mY$ of $G$.

  \item $\C{X}{Y} :\equiv f_X(X) \equiv f_Y(Y)$\;.

\end{itemize}
\end{definition}

\section{Impossibility Results for Classical Secure Function Evaluation}\label{sec:sfe-lower-bounds}

Let a protocol be an $\eps$-secure implementation of a primitive $f:\mX \times \mY \rightarrow \mZ$ in the semi-honest model. Let $P_{XY}$ be the input distribution % and let $P_{\bar X \bar Y}$ be the corresponding output distribution of the ideal primitive, i.e., $P_{\bar X \bar Y}:=P_{XY}P_{\bar X \bar Y\mid XY}$, 
and let $M$ be the whole communication during the execution of the protocol. Then the security of the protocol implies the following lemma that we will use in our proofs.
\begin{lemma}
\label{lem:entropyLowerBound}
\[\cHS{X}{VM} \geq \cHS{X}{Yf(X,Y)}-\eps \log |\mathcal{X}|- h(\eps)\;.\]
\end{lemma}
\begin{proof}
The security of the protocol implies that there exists a randomized function $S_B$, the simulator, such that $\dis(P_{X Y S_B(Y,f(X,Y))},P_{X Y VM})\leq \eps$. We can use
Lemma \ref{lem:h-smooth} and \eqref{eq:H-markov} to obtain
\begin{align*}
\cHS{X}{VM} &\geq \cHS{X}{S_B(Y,f(X,Y))}-\eps \log |\mathcal{X}|- h(\eps)\\
&\geq \cHS{X}{ Y f(X,Y)}-\eps \log |\mathcal{X}|- h(\eps)\;.
\end{align*}

\end{proof}

We will now give lower bounds for information-theoretically secure implementations of functions $f:\mX \times \mY \rightarrow \mZ$ from a primitive $P_{UV}$ in the semi-honest model. Let $f:\mX \times \mY \rightarrow \mZ$ be a function such that
\begin{align} \label{nonTrivialF-1}
\forall x\neq x' \in \mX\;  \exists y \in \mY: \; f(x,y) \neq f(x',y)\;.
\end{align}

This means that it is possible to compute $x$ from the set $\{(f(x,y),y):~y\in\mY\}$ for any $x$. In any secure implementation of $f$, Alice does not learn which $y$ Bob has chosen, but has to make sure that Bob can compute $f(x,y)$ for any $y$. This implies that she cannot hold back any information about $x$. The statement of Lemma \ref{lem:entropyUpperBound-n} formally captures this intuition. 

Unless otherwise specified, we assume that Alice and Bob choose their inputs $X$ and $Y$ uniformly at random  and independent of everything else in the following.
\begin{lemma}\label{lem:entropyUpperBound-n}
For any protocol that is an $\eps$-secure implementation of a function $f:\mX \times \mY \rightarrow \mZ$ that satisfies \eqref{nonTrivialF-1} in the semi-honest model, we have for any~$y \in \mY$
\begin{align*} %\label{eq:fano}
 % H(X \mid UM) \leq  (3|\mY|-2)(\eps\log|\mZ| + h(\eps))
   \cHS{X}{UM,Y=y} \leq  (3|\mY|-2)(\eps\log|\mZ| + h(\eps))\;
\end{align*}
\end{lemma}
\begin{proof}
There exists a randomized function $S_A$, the simulator, such that \[\dis (P_{XMU\mid Y=y},P_{X S_A(X)}) \leq \eps\] for all $y \in \mY$. Therefore, the triangle inequality implies that for any $y,y'$
\begin{align} \label{lemma:entropyUpperBound:dist}
\dis (P_{XMU\mid Y=y},P_{XMU\mid Y =y'}) \leq 2\eps \;.
\end{align}
If Bob's input is fixed to any $y\in \mY$ and Alice's input $X$ is still uniform, then the resulting distribution of $XZUM$ is $P_{XZUM|Y=y}$. Thus, we can apply the Markov chain relation~\eqref{ineq:markov_second} and the monotonicity of the conditional mutual information to obtain the equality $\cIS{X}{Z}{UM,Y=y}=0$ for all $y$. Furthermore, we have $\Pr[Z \neq f(X,Y) \mid Y=y]\leq \eps$. Thus, it follows from \eqref{eq:H-markov} and \eqref{eq:fano} that 
\begin{align}\label{ineq:entropyUpperBound}
\cHS{f(X,y)}{ UM, Y=y}&\leq \cHS{f(X,y) }{ Z, Y=y}\nonumber\\
& \leq \eps\cdot\log|\mZ|+h(\eps) \;.
\end{align}
Together with (\ref{lemma:entropyUpperBound:dist}) and Lemma \ref{lem:h-smooth} %(\ref{eq:H-smooth}), 
this implies that for any $y,y'$
\begin{align*}
\cHS{f(X,y)}{  UM, Y=y'}
&\leq 3\eps \log|\mZ| + h(\eps) + h(2\eps)\\
&\leq 3(\eps \log|\mZ|+ h(\eps))\;,
\end{align*}
where the second inequality follows from (\ref{eq:h}).
Since $X$ can be computed from the values $f(X,y_1), \dots, f(X,y_{|\mY|})$, we obtain
\begin{align*}\cHS{X}{ UM,Y=y}&\leq \cHS{f(X,y_1), \dots f(X,y_{|\mY|})}{ UM,Y=y} \\
 &\leq \sum_{y' \in \mY} \cHS{f(X,y')}{ UM,Y=y}  \\
 &\leq (3|\mY|-2)(\eps\log|\mZ|  + h(\eps))\;,
\end{align*}
where we used (\ref{eq:mono}) in the first and (\ref{eq:chain}) and (\ref{eq:mono}) in the second inequality.
%%\qed
\end{proof}

\begin{theorem} \label{thm:H-imposs-n}
Let $f:\mX \times \mY \rightarrow \mZ$ be a function that satisfies \eqref{nonTrivialF-1}. If there exists a protocol that implements $f$ from a primitive $P_{UV}$ with an error $\eps$ in the semi-honest model, then
\begin{align*}
\cHS{U}{V} \geq &\max_{y}\cHS{X}{f(X,y)}\\
&~-(3|\mY|-1)(\eps\log|\mZ| + h(\eps))-\eps \log |\mathcal{X}|\;.
\end{align*}
\end{theorem}
\begin{proof}
Let $y \in \mY$. By Lemma~\ref{lem:entropyUpperBound-n} and inequality~\eqref{eq:mono}, we conclude that
\begin{align*}
 \cHS{X}{ UVM,Y=y}&\leq \cHS{X}{ UM,Y=y}\\
&\leq (3|\mY|-2)(\eps \log|\mZ| + h(\eps))\;.
\end{align*}
We can use \eqref{eq:mono}, \eqref{eq:chain} and Lemma \ref{lem:h-smooth} to obtain
\begin{align*}
%\cHS{X \mid f(X,y))-\eps \log(|\mathcal{X}|)- h(\eps)&\leq 
&\cHS{X}{VM,Y=y}\\
&~~~=\cHS{U}{VM,Y=y}+\cHS{X}{UVM,Y=y}\\
&~~~~~- \cHS{U}{XVM,Y=y}\\
&~~~\leq\cHS{U}{VM,Y=y}+(3|\mY|-2)(\eps\log|\mZ| + h(\eps))\\
&~~~\leq\cHS{U}{V} +(3|\mY|-2)(\eps\log|\mZ| + h(\eps))\;.
\end{align*}
By applying Lemma~\ref{lem:entropyLowerBound} to the case where Alice's input $X$ is uniform and Bob's input is fixed to $y$, we know that 
\begin{align*}
\cHS{X}{f(X,y)}-\eps \log |\mathcal{X}|- h(\eps)&\leq \cHS{X}{VM,Y=y}\;.
\end{align*}
The statement follows by maximizing over all $y$. 
%\qed
\end{proof}

Note that in~\eqref{ineq:entropyUpperBound} the term $\log |\mZ|$ could be replaced by \[d_f:=\log \max_y |\{f(x,y):~x \in \mX\}|\leq \log \min(|\mZ|,|\mX|).\] The resulting bound,
\begin{align*}
\cHS{U}{V} \geq &\max_{y}\cHS{X}{f(X,y)}\\&-(3|\mY|-1)(\eps\cdot d_f + h(\eps))-\eps \log |\mathcal{X}|\;,
\end{align*}
is stronger in general, but does not lead to improved results for the examples considered here. 

If the domain $|\mY|$ of a function is large, Theorem~\ref{thm:H-imposs-n} may only imply a rather weak bound. A simple way to improve this bound is to restrict the domain of $f$, i.e., to consider a function $f'(x,y): \mX' \times \mY' \rightarrow \mZ$ where $\mX' \subset \mX$ and $\mY' \subset \mY$ with $f'(x,y) = f(x,y)$ that still satisfies condition \eqref{nonTrivialF-1}. Clearly, if $f$ can be computed from a primitive $P_{UV}$ with an error $\eps$ in the semi-honest model, then $f'$ can be computed with the same error. Thus, any lower bound for $f'$ implies a lower bound for $f$.

\cancel{
The following corollaries for $\OT{t}{n}{k}$, $\EQ{n}$ and $\IP{n}$ follow immediately from Theorem \ref{thm:H-imposs-n}. 
}

\begin{corollary} \label{cor:H-imposs-n}
For any implementation of $m$ independent instances of $\OT{t}{n}{k}$ from a primitive $P_{UV}$ that is $\eps$-secure in the semi-honest model, the following lower bound must hold:
\begin{align*}
\cHS{U}{V} &\geq ((1-\eps)n-t)km -\left (3\lceil n/t \rceil -1\right )(\eps mtk+h(\eps))\;.
\end{align*}
\cancel{
\severin{\begin{align*}
\cHS{U}{V} &\geq (n-t)k-\left (3\lceil n/t \rceil -2\right )(\eps tk+h(\eps))\\
&~~-\eps nk-h(\eps)\\
&\geq (n-t)k-\left (4n+t)\eps k-(3\lceil n/t \rceil -1\right )h(\eps)\\
&\geq (n-t)k- 5kn\eps-(3\lceil n/t \rceil -1)h(\eps)\\
&\geq (n-t)k- 5kn\eps-3nh(\eps)\;.
\end{align*}}
}
\end{corollary}
\begin{proof}
We can choose subsets $C_i \subseteq \{0,\dots,n-1\}, \text{ with } 1\leq i \leq \lceil n/t \rceil$, of size $t$ such that $\bigcup_{i=1}^{\lceil n/t \rceil}C_i=\{0,\dots,n-1\}$, and restrict Bob to choose one of these sets as input for every instance of OT. It is easy to check that condition \eqref{nonTrivialF-1} is satisfied. %Then the matrix representing the restriction of $\OT{t}{n}{k}$ to $\{0,1\}^{nk} \times \bigcup_{i=1}^{\lceil n/t \rceil}C_i$ has no two identical rows. 
The statement follows from Theorem \ref{thm:H-imposs-n}. 
\end{proof}

For our next lower-bound, the function $f$ must satisfy the following property.
 Let $f:\mX \times \mY \rightarrow \mZ$ be a function such that there exist $y_1 \in \mY$ such that
\begin{align} \label{conditionY1}
 \forall x \neq x' \in \mX: f(x,y_1) \neq f(x',y_1)\;,
\end{align}
and $y_2 \in \mY$ such that
\begin{align} \label{conditionY2}
 \forall x,x' \in \mX: f(x,y_2)=f(x',y_2)\;.
\end{align}
Therefore, Bob will receive Alice's whole input if his input is $y_1$, and will get no information about Alice's input if his input is $y_2$. This property can for example be satisfied by restricting Alice's input in $\OT{t}{n}{k}$, as we will see in Corollary~\ref{cor:I-imposs-n}.

%Let Alice's and Bob's inputs $X$ and $Y$ be independent and uniformly distributed. Then the following lemma holds for any protocol that implements $f$ from a primitive $P_{UV}$ with an error of at most $\eps$ in the semi-honest model. The following lemma states that if it is possible that Bob does not receive the input $x$ (by choosing input $y_2$), then $M, \C{U}{V}$ should not reveal $x$, even if Bob receives it (by choosing $y_1$).
Let Alice's input $X$ be uniformly distributed. Loosely speaking, the security of the protocol implies that the communication gives (almost) no information about Alice's input $X$ if Bob's input is $y_2$. But the communication must be (almost) independent of Bob's input, otherwise Alice could learn Bob's input. Thus, Alice's input $X$ is uniform with respect to the whole communication even when Bob's input is $y_1$. Let now Bob's input be fixed to $y_1$ and let $M$ be the whole communication. The following lower bound can be proved using the given intuition.  

\begin{lemma}\label{lem:almostUniform-n}
\begin{align*}
H(f(X,y_1)|M,\C{U}{V}&,Y=y_1)\\
&\geq \log|\mX| - 6(\eps \log|\mX|+h(\eps))\;.
\end{align*}
\end{lemma}
\begin{proof}
Let $g_U,g_V$ be the functions that compute the common part of $P_{UV}$. As in inequality~\eqref{lemma:entropyUpperBound:dist} in the proof of Lemma~\ref{lem:entropyUpperBound-n}, we obtain that  
\[\dis(P_{XMU\mid Y=y},P_{XMU\mid Y=y'})\leq 2\eps \;,\]
for all $y\neq y' \in \mY$. This implies that 
\begin{align} \label{eq:upper1}
\dis(P_{XMg_U(U)\mid Y=y},P_{XMg_U(U)\mid Y=y'})\leq 2\eps \;,
\end{align}
and 
\begin{align} \label{eq:upper2}
\dis(P_{X}P_{ Mg_U(U)\mid Y=y},P_{X}P_{Mg_U(U) \mid Y=y'})\leq 2\eps \;.
\end{align}
Since the protocol is secure, there exists a simulator $S_B$ such that
\[
 \dis(P_{XM V\mid Y=y_2},P_{X S_B(y_2,f(X,y_2)) })\leq \eps\;.
\]
From the property \eqref{conditionY2}, we can conclude that
\[\dis(P_{XM V \mid Y=y_2},P_X P_{S_B(y_2,f(X,y_2))})\leq \eps.\]
From security of the protocol (Definition  \ref{definition:sec}) we know that $\dis(P_{S_B(y_2,f(X,y_2))},P_{M V \mid Y=y_2})\leq\eps$, which immediately implies $\dis(P_X P_{S_B(y_2,f(X,y_2))},P_XP_{M V \mid Y=y_2})\leq\eps$. Therefore, we can use the triangle inequality to derive the following upper bound on the distance from uniform of $X$ with respect to  $M g_U(U)$ conditioned on $y_2$:
\begin{align}
\dis(P_{X M g_U(U) \mid Y=y_2}&,P_X P_{M g_U(U) \mid Y=y_2})\nonumber\\
\leq&\;  \dis(P_{X M V \mid Y=y_2},P_X P_{M V \mid Y=y_2}) \nonumber\\
\leq&\;   \dis(P_{XM V \mid Y=y_2},P_X P_{S_B(y_2,f(X,y_2))}) \nonumber\\
&+ \dis(P_X P_{S_B(y_2,f(X,y_2))},P_XP_{M V \mid Y=y_2}) \nonumber\\
 \leq& \;2\eps\;.  \label{eq:upper3}
\end{align}
This implies that a weaker upper bound also holds conditioned on $y_1$ as follows: We can use the triangle inequality again to conclude from \eqref{eq:upper1}, \eqref{eq:upper2} and \eqref{eq:upper3} that
\begin{align*}
\dis(P&_{X Mg_U(U)\mid Y=y_1},P_{X}P_{Mg_U(U) \mid Y=y_1})\\
&~~~~~~\leq  \dis(P_{X Mg_U(U)\mid Y=y_1},P_{X Mg_U(U)\mid Y=y_2})\\
&~~~~~~~~ + \dis(P_{X Mg_U(U)\mid Y=y_2},P_X P_{Mg_U(U)\mid Y=y_2})\\
&~~~~~~~~ + \dis(P_X P_{Mg_U(U)\mid Y=y_2},P_{X} P_{Mg_U(U) \mid Y=y_1})\\
&~~~~~~\leq  \; 6\eps \;.
\end{align*}
Therefore, we obtain 
\begin{align*}
\Hop(f(X,y_1) | M,\C{U}{V},&Y =y_1)\\
&=\Hshannon{X}{  M,\C{U}{V},Y =y_1}\\
%&\geq \log|\mX|-6\eps\log|\mX| -h(6\eps)\\
&\geq \log|\mX| -6(\eps\log|\mX|- h(\eps))\;,
\end{align*} 
where we used Lemma~\ref{lem:h-smooth}. 
%\qed
\end{proof}

We use Lemma~\ref{lem:almostUniform-n} to prove the following lower bound on the mutual information of the distributed randomness for implementations of a two-party function $f$ from a primitive $P_{UV}$ in the semi-honest model.%reductions of $\OT{1}{n}{k}$ in the semi-honest model.% The proof is given in Appendix \ref{app:lem:almostUniform-n}.

\begin{theorem} \label{thm:I-imposs-n}
Let $f:\mX \times \mY \rightarrow \mZ$ be a function that satisfies \eqref{conditionY1} and \eqref{conditionY2}. 
%such that there exist $y_1,y_2 \in \mY$ with $f(x,y_1)\neq f(x',y_1)$ and $f(x,y_2)=f(x',y_2)$ for all $x,x' \in \mX$.
Then, for any protocol that implements $f$ with an error of at most $\eps$ in the semi-honest model from a primitive $P_{UV}$, the following lower bound must hold:
\begin{align*}
\IS{U}{V}&\geq\cIS{U}{V}{\C{U}{V}}\\
&\geq \log|\mX|-7(\eps \log|\mX|+ h(\eps))\;.
\end{align*}
\end{theorem}
\begin{proof}
Let Alice's input $X$ be uniformly distributed and Bob's input be fixed to $y_1$. Let $Z$ be Bob's output and $M$ the 
whole communication. Then Lemma~\ref{lem:almostUniform-n} implies that 
\begin{align}\label{eq:thm3-1}
 \cHS{f(X,y_1)}{M,\C{U}{V}}\geq\log|\mX| -6(\eps\log|\mX| - h(\eps))\;.
\end{align}
Since $\Pr[Z \neq f(X,y_1)] \leq \eps$ and $\Markov{X}{VM}{Z}$, it follows from
(\ref{eq:H-markov}) and (\ref{eq:fano}) that
\begin{align}\label{eq:thm3-2}
\cHS{f(X,y_1)}{V M} 
 \leq \cHS{f(X,y_1)}{Z}
 \leq \eps \log|\mX| + h(\eps) \;.
\end{align}
Inequalities \eqref{eq:thm3-1} and \eqref{eq:thm3-2} imply, using $\Markov{X}{UM}{ZYV}$, \eqref{eq:I-markov2} and \eqref{eq:mono-I}, that
\begin{align*}
\cIS{U}{V}{M,\C{U}{V}}
&\geq\cIS{X}{V}{ M,\C{U}{V}}\\
&\geq\cIS{f(X,y_1)}{V}{ M,\C{U}{V}}\\
&=\cHS{f(X,y_1)}{ M,\C{U}{V}}\\
&~~~-\cHS{f(X,y_1)}{VM,\C{U}{V}}\\
&\geq \log|\mX|-7(\eps \log|\mX|- h(\eps))\;. 
\end{align*}
Let $M^i := (M_1, \dots, M_i)$, i.e., the sequence of all messages sent until the $i$-th round. Without loss of generality, let us assume that Alice sends the message of the $(i+1)$-th round. Since, we have $\Markov{M^{i+1}}{M^i U}{V}$, it follows from (\ref{eq:I-markov}) that
\[\cIS{U}{V}{ M^{i+1},\C{U}{V}}  \leq \cIS{U}{V}{ M^i,\C{U}{V}}\;.\]
By induction over all rounds, it holds that
\[\cIS{U}{V}{M,\C{U}{V}}  \leq \cIS{U}{V}{\C{U}{V}}\;.\]
Since $\Markov{\C{U}{V}}{U}{V}$, the statement of the theorem follows.
%\qed
\end{proof}

The next corollary provides a lower bound on the mutual information for implementations of $\OT{t}{n}{k}$ from a primitive $P_{UV}$. It follows immediately from Theorem \ref{thm:I-imposs-n}.
\begin{corollary}\label{cor:I-imposs-n}
If there exists a  protocol that implements $m$ independent instances of $\OT{t}{n}{k}$  from a primitive $P_{UV}$ with an error of at most $\eps$ in the semi-honest model, then the following lower bounds must hold: If $t \leq \lfloor n/2 \rfloor$, then
\begin{align*}
\cIS{U}{V}{\C{U}{V}} \geq mtk-7(\eps mtk + h(\eps))\;.
\end{align*}
If $t > \lfloor n/2 \rfloor$, then
\begin{align*}
\cIS{U}{V}{\C{U}{V}} \geq m(n-t)k-7(\eps m(n-t)k + h(\eps))\;.
\end{align*}
\end{corollary}
\begin{proof}
In the first case, consider the function that is obtained by setting the first $n-t$ inputs to a fixed value and choosing the remaining $t$ inputs from $\{0,1\}^{tk}$ for every instance of OT. In the second case, we use the fact that $\OT{n-t}{2n-2t}{k}$ can be obtained from $\OT{t}{n}{k}$ by fixing $2t-n$ inputs. Thus, both bounds follow from Theorem \ref{thm:I-imposs-n}.
%\qed
\end{proof}

%\subsection{Additional Bounds for Implementations of OT}\label{subsec:add-bound}
An instance of $\OT{1}{2}{1}$ can be implemented from one instance of $\OT{1}{2}{1}$ in the opposite direction \cite{WolWul06}. Therefore, 
it follows immediately from Corollary~\ref{cor:H-imposs-n} that 
\begin{align*}
\cHS{V}{U} \geq 1 - 5h(\eps)-7\eps\;,
\end{align*}
since any violation of this bound would contradict the bound of Corollary~\ref{cor:H-imposs-n}. We will show that a generalization of this bound also holds for $m$ independent copies of $\OT{1}{n}{k}$ for any $n \geq 2$. Note that we can assume that $k=1$. The resulting bound then also implies a bound for $k>1$ because one instance of $\OT{1}{n}{1}$ can be  implemented from one instance of $\OT{1}{n}{k}$. 

\begin{theorem} \label{thm:H-imposs-c}
Let a protocol having access to $P_{UV}$ be an $\eps$-secure implementation of $m$ independent copies of 
$\OT{1}{n}{1}$ in the semi-honest model. Then 
\begin{align*}
\cHS{V}{U} & \geq m\log n - m(4\log n + 7 )(\eps + h(\eps))\;.
\end{align*}
\end{theorem}

\begin{proof}
Let Alice and Bob choose their inputs $X=(X^1,\dots,X^m)\in\{0,1\}^{mn}$, where $X^i=(X^i_0,\ldots,X^i_{n-1})$, and $C=(C^1,\dots,C^m) \in \{0, \dots, n-1\}^m$ uniformly at random. Let $Y=(Y^1,\dots,Y^m)$ be the output of Bob at the end of the protocol. Let $j \in \{1,\dots,m\}$. First, we consider the  $j$th instance of $\OT{1}{n}{1}$. Let $A_i := X^j_0 \oplus X^j_i$, for $i \in \{1, \dots, n-1\}$. From the security of the protocol follows that there exists a randomized function $S_B(c,x_c)$ such that for all $a = (a_1,\dots,a_{n-1}) \in \{0,1\}^{n-1}$, 
\[ \dis(P_{YCVM \mid A=a},P_{X_C C S_B(C,X_C)}) \leq \eps\;.\]
Hence, the triangle inequality implies that 
\begin{align} \label{eq:eq1}
 \dis(P_{Y^jC^jVM \mid A=a},&P_{Y^jC^jVM \mid A=a'})\nonumber\\
&\leq\dis(P_{YCVM \mid A=a},P_{YCVM \mid A=a'})\nonumber\\
 &\leq 2\eps\;
\end{align}
holds for all $a, a'$. We have $\Pr[Y^j \neq X^j_C \mid A=a] \leq \eps$ for all $a$. 
If $A = (0, \dots, 0)$, 
we have
$X^j_C = X^j_0$. Since $\Markov{X^j}{VM}{Y^j}$, it follows from (\ref{eq:mono}) and (\ref{eq:fano})  that
\begin{align} \label{eq:H-inf:Y-bound}
 \cHS{Y^j}{VM, A=(0, \dots, 0)}&\leq \cHS{Y^j}{X^j, A=(0, \dots, 0)}\nonumber \\
 &\leq \cHS{Y^j }{ X^j_0, A=(0, \dots, 0)}\nonumber\\
 &\leq \eps + h(\eps)\;.
\end{align}
Now, we map $C^j$ to a bit string of size $\lceil \log n \rceil$. Let $C_b$ be the $b$-th bit of that bit string, where $b \in \{0, \dots, \lceil \log n \rceil - 1\}$. Let $a^b = (a^b_1, \dots, a^b_{n-1})$, where $a^b_i = 1$ if and only if the $b$-th bit of the binary representation of $i$ is $1$. Conditioned on $A = a^b$, we have $X^j_C = X^j_0 \oplus C_b$. It follows from $\Markov{X^j}{VM}{Y^jC^j}$, (\ref{eq:mono}) and (\ref{eq:fano})  that
\begin{align}\label{eq:H-inf:YC-bound}
\cHS{Y^j \oplus C_b}{ VM, A=a^b} &\leq \cHS{Y^j \oplus C_b}{ X^j_0, A=a^b}\nonumber \\
& \leq \eps + h(\eps)\;.
\end{align}
By Lemma \ref{lem:h-smooth}, \eqref{eq:eq1} and \eqref{eq:H-inf:Y-bound}, we obtain
\[\cHS{Y^j}{VMA} \leq \eps + h(\eps) + 2\eps + h(2\eps) \leq 3 \eps + 3 h(\eps).\]
It follows from Lemma \ref{lem:h-smooth}, \eqref{eq:eq1} and \eqref{eq:H-inf:YC-bound} that for all $b$
\[\cHS{Y^j \oplus C_b}{VMA} \leq 3 \eps + 3 h(\eps) \;.\]
Since $(C^j,Y^j)$ can be calculated from $(Y^j,Y^j \oplus C_0, \dots, Y^j \oplus C_{\lceil \log n \rceil - 1})$, this implies that
\[\cHS{C^j Y^j}{VMA} \leq 3 (\lceil \log n \rceil + 1 )(\eps + h(\eps)) \;.\]
The Markov chain $\Markov{A}{VM}{C^jY^j}$, $\lceil \log n \rceil \leq \log n + 1$  and inequality \eqref{eq:mono} imply that
\[\cHS{C^j}{VM} \leq 3 (\log n + 2 )(\eps + h(\eps))\;.\]
Thus we can use \eqref{eq:chain} and \eqref{eq:mono} to obtain
\begin{align*}
\cHS{C}{VM} &\leq \sum_{j=1}^n \cHS{C^j}{VM}\\
&\leq 3m(\log n + 2 )(\eps + h(\eps))\;. 
\end{align*}
We can use \eqref{eq:chain}, \eqref{eq:mono} and Lemmas~\ref{lem:h-smooth}  to obtain 
\begin{align*}
%m(\log n-\eps \log n)& - h(\eps)\\
\cHS{C}{UM}&=\cHS{V}{UM}+\cHS{C}{UVM} - \cHS{V}{CUM}\\
&\leq \cHS{V}{UM}+3m(\log n + 2 )(\eps + h(\eps))\\
&\leq \cHS{V}{U}+3m(\log n + 2 )(\eps + h(\eps))\;.
\end{align*}
%Lemma~\ref{lem:entropyLowerBound}~implies 
The security of the protocol implies that there exists a randomized function $S_A$ such that $\dis(P_{CS_A(X)},P_{CUM})\leq \eps$. Using
Lemma \ref{lem:h-smooth} and inequality \eqref{eq:H-markov}, we obtain that
\begin{align*}
\cHS{C}{UM} &\geq \cHS{C}{S_A(X)}-\eps m\log n - h(\eps)\\
&\geq \cHS{C}{X}-\eps m\log n- h(\eps)
\end{align*}
%\[\cHS{C}{UM}\geq m(\log n-\eps \log n) - h(\eps) \;.\]
%\qed
\end{proof}

Altogether, Corollary \ref{cor:H-imposs-n}, Corollary~\ref{cor:I-imposs-n} and Theorem~\ref{thm:H-imposs-c} prove the following theorem. 
\begin{theorem}\label{thm:impossibility}
If there exists a protocol having access to $P_{UV}$ that implements $m$ instances of 
$\OT{1}{n}{k}$ with an error of at most $\eps$ in the semi-honest model, then 
\begin{align*}
\cHS{U}{V}& \geq m (n-1) k - (4n-1) (\eps mk + h(\eps))\;, \\
\cHS{V}{U} & \geq m\log n - m(4\log n + 7 )(\eps + h(\eps))\;,\\
\cIS{U}{V}{\C{U}{V}} & \geq m k-7\eps m k-7 h(\eps)\;.
\end{align*}
\end{theorem}

Since $m$ instances of $\OT{1}{n}{k}$ are equivalent to a primitive $P_{UV}$ with $\cHS{U}{V} = m (n-1) k$, $\IS{U}{V} = m k$ and  $\cHS{V}{U} = m \log n$, any protocol that implements $M$ instances of $\OT{1}{N}{K}$ from $m$ instances of $\OT{1}{n}{k}$ with an error of at most $\eps$ needs to satisfy the following inequalities:
\begin{align*}
m (n-1) k &\geq M(N-1) K - (4N-1) (\eps M K + h(\eps))\;,\\
m k &\geq MK-7\eps MK-7 h(\eps)\;,\\
m \log n &\geq M\log N - M(4\log N + 7 )(\eps + h(\eps))\;.
\end{align*}
Thus, we get Corollary~\ref{cor:main2}.

We will now use the proof of Theorem~\ref{thm:H-imposs-n} and the smooth entropy formalism to derive a lower bound on the conditional min-entropy for information-theoretically secure implementations of functions $f:\mX \times \mY \rightarrow \mZ$ from a primitive $P_{UV}$ in the semi-honest model. As a motivation, consider the following question: is it possible to $\eps$-securely implement $\OT{1}{2}{K}$ from  $\RabinOT{1/2}{k}$? Corollary \ref{cor:H-imposs-n} only tells us that $K$ must be smaller than or equal to $k/2$. Our lower bound on the conditional smooth min-entropy, however, implies that there is no such implementation if $K\geq 2$ and $0\leq \eps<0.25$, independently of $k$. 

Let $f:\mX \times \mY \rightarrow \mZ$ be a function that satisfies~\eqref{nonTrivialF-1}. Let Alice and Bob choose their inputs $X$ and $Y$ uniformly at random and let $M$ be the whole communication during the protocol. For the rest of this section, we assume that all parameters are sufficiently small such that the smoothing parameters of the smooth entropies are always in $[0,1)$.

\begin{lemma}\label{lem:max-entropyUpperBound-n}
If there exists an $\eps$-secure implementation of $f:\mX \times \mY \rightarrow \mZ$ from a primitive $P_{UV}$ in the (weak) semi-honest model, then
\begin{align*} %\label{eq:fano}
 % H(X \mid UM) \leq  (3|\mY|-2)(\eps\log|\mZ| + h(\eps))
   \Chmaxeps{3|\mY|\eps}{X}{UM,Y=y} =0\;.
\end{align*}
\end{lemma}

\begin{proof}
Since the protocol is secure for Bob, there exists a randomized function $S_A$ such that \[\dis(P_{XMU\mid Y=y},P_{X S_A(X)}) \leq \eps\] for all $y \in \mY$.  Therefore, for any $y,y'$
\begin{align} \label{lemma:max-entropyUpperBound:dist}
\dis (P_{XMU\mid Y=y},P_{XMU\mid Y =y'}) \leq 2\eps\;.
\end{align}
It holds that $\cIS{X}{Z}{UM,Y=y}=0$. Furthermore, we have $\Pr[Z \neq f(X,Y) \mid Y=y]\leq \eps$. Thus, Lemmas~\ref{lem:markov-max} and \ref{lem:mono-max} imply that
\begin{align}
\Chmaxeps{\eps}{f(X,y)}{ UM, Y=y}&\leq \Chmaxeps{\eps}{f(X,y)}{ Z, Y=y}= 0.
\end{align}
Together with \eqref{lemma:max-entropyUpperBound:dist}, this implies that for any $y,y'$
\begin{align*}
\Chmaxeps{3\eps}{f(X,y)}{UM, Y=y'}=0\;.
\end{align*}
Since $X$ can be computed from the values $f(X,y_1), \dots, f(X,y_{|\mY|})$, we obtain
\begin{align*}
&\Chmaxeps{3|\mY|\eps}{X}{ UM,Y=y}\\
&~~~~~~~~~~~~~\leq \Chmaxeps{3|\mY|\eps}{f(X,y_1), \dots ,f(X,y_{|\mY|})}{ UM,Y=y} \\
&~~~~~~~~~~~~~ \leq \sum_{y' \in \mY} \Chmaxeps{3\eps}{f(X,y')}{UM,Y=y}  \\
 &~~~~~~~~~~~~~=0\;.
\end{align*}
where we used Lemma \ref{lem:mono-max} and the subadditivity of the max-entropy (Lemma~\ref{lem:sub-add}).
%\qed
\end{proof}

Let $P_{XY}$ be the input distribution to the ideal primitive. Then the security of the protocol implies the following lemma. %that we will use in our proofs.
\begin{lemma}
For any protocol that is an $\eps$-secure implementation of $f:\mX \times \mY \rightarrow \mZ$ from a primitive $P_{UV}$ in the semi-honest model,
\label{lem:min-entropyLowerBound}
\[\Chminteps{\eps+\eps'}{X}{VM} \geq \Chminteps{\eps'}{X}{Y f(X,Y)}\;,\]
for any $\eps'\geq 0$.% such that $\eps+\eps'<1$
\end{lemma}
\begin{proof}
The security of the protocol implies that there exists a randomized function $S_B$, the simulator, such that $\dis(P_{X Y S_B(Y,f(X,Y))},P_{X Y VM})\leq \eps$. Therefore, we obtain
\begin{align*}
\Chminteps{\eps+\eps'}{X}{VM} &\geq \Chminteps{\eps'}{X}{S_B(Y,f(X,Y))}\\
&\geq \Chminteps{\eps'}{X }{ Y f(X,Y)}\;,
%&\geq \min_y \Hmin(X \mid Y f(X,y))\;.
\end{align*}
where we used Lemma~\ref{lem:markov-t} in the second inequality.
%\qed
\end{proof}

\begin{theorem} \label{thm:Hmin-imposs-n}
Let $f:\mX \times \mY \rightarrow \mZ$ be a function that satisfies \eqref{nonTrivialF-1}. If there exists a protocol having access to a primitive $P_{UV}$ that implements $f$ with an error of at most $\eps$ in the semi-honest model, then
\begin{align*}
\Chminteps{(3|\mY|+1)\eps+\eps'}{U}{V} \geq \max_{y}\Chminteps{\eps'}{X}{f(X,y)}\;,
\end{align*}
for any $\eps'\geq 0$.
\end{theorem}
\begin{proof}
Let $y \in \mY$. It follows from Lemmas~\ref{lem:max-entropyUpperBound-n} and \ref{lem:mono2} that
\[\Hmax^{3|\mY|\eps}(X|UVM,Y=y)\leq \Hmax^{3|\mY|\eps}(X|UM,Y=y)=0\;.\]
Therefore, Lemma \ref{lem:mono2} and \ref{lem:mutual-inf2} implies that 
\begin{align*}
\Chminteps{\eps+\eps'}{X}{VM,Y=y}&-\Chmaxeps{3|\mY|\eps}{X}{UVM,Y=y}\\&\leq \Chminteps{(3|\mY|+1)\eps+\eps'}{U}{VM,Y=y}\\
&\leq \Chminteps{(3|\mY|+1)\eps+\eps'}{U}{V}\;.
\end{align*}
We can use Lemma~\ref{lem:min-entropyLowerBound} to obtain
\begin{align*}
\Chminteps{\eps'}{X}{f(X,y)}&\leq \Chminteps{\eps+\eps'}{X}{ VM,Y=y}\;.
\end{align*}
The statement follows by maximizing over all $y$. 
%\qed
\end{proof}
\cancel{
Note that the proof holds for any definition of the smoothed conditional min-entropy such that Lemmas~\ref{lem:mutual-inf},~\ref{lem:mono} and~\ref{lem:markov} hold. Thus, Lemmas~\ref{lem:mutual-inf2},~\ref{lem:mono2} and~\ref{lem:markov2} imply the following corollary.
\begin{corollary} \label{cor:Hmin-imposs-n}
Let $f:\mX \times \mY \rightarrow \mZ$ be a function that satisfies \eqref{nonTrivialF-1}. Let a protocol having access to $P_{UV}$ be an $\eps$-secure implementation of $f$ in the semi-honest model. Then
\begin{align*}
\Chminteps{(3|\mY|+1)\eps}{U}{V} \geq \max_{y}\CHmint(X \mid f(X,y))\;.
\end{align*}
\end{corollary}
}

\subsection{Lower Bounds for Protocols implementing OT}\label{lower-bounds:ot}

\begin{corollary}\label{cor:ot-imposs}
Any protocol that implements $M$ instances of $\OT{1}{N}{K}$ from  $m$ instances of $\OT{1}{n}{k}$ with an error of at most $0\leq \eps < \frac{1}{2(3n+1)}$ in the semi-honest model must satisfy
\cancel{
\begin{align*}
 m(n-1)k\geq M(N-1)K-\log\left(\frac{1}{1-(3n+1)\eps}\right)\;.
\end{align*}
In particular, if $0\leq \eps \leq \frac{1}{2(3n+1)}$, then 
}
\begin{align*}
 m(n-1)k\geq M(N-1)K-(6n+2)\eps\;.
\end{align*}
%\severin{den ersten Bound kann man jeweils auch einfach weglassen: das schraenkt den bereich von $\eps$ ja nicht stark ein.}
\end{corollary}

\begin{proof}
From Theorem \ref{thm:Hmin-imposs-n} follows that 
\begin{align}\label{eq:min-ent-low}
 \Chminteps{(3n+1)\eps}{U}{V} \geq M(N-1)K\;.
\end{align}
For the distribution $P_{UV}$ of randomized OTs, the entropy $\Chminteps{\bar \eps}{U}{V}$ with $0\leq \bar \eps < 1$ is maximized by the event $\Omega$ with $P_{\Omega|U=u,V=v}=1-\bar \eps$ for all $u,v$ in the support of $P_{UV}$. Therefore, we have 
\begin{align}\label{eq:min-ent-up}
\Chminteps{(3n+1)\eps}{U}{V}&\leq -\log(2^{-m(n-1)k}(1-(3n+1)\eps))\nonumber\\
&= m(n-1)k-\log(1-(3n+1)\eps)\;.
\end{align}
The statement follows from the fact that $\log(1/\eps)\leq 2(1-\eps)$ for $1/2\leq\eps\leq 1$.
%\qed
\end{proof}

This corollary implies that there is no protocol that extends $\OT{1}{2}{1}$ in the semi-honest model.
\begin{corollary}
Any protocol that implements $m+1$ instances of $\OT{1}{2}{1}$ in the semi-honest model using $m$ instances of $\OT{1}{2}{1}$ must have an error 
$ \eps \geq 1/14.$
\end{corollary}

%\severin{ein etwas sinnvolleres beipiel bekommt man mit $\Chminteps{\eps}{U}{V}$: For $\eps<2^{-m}$ we have $\Chminteps{\eps}{U}{V}=0$ if $(U,V)$ corresponds to $m$ instances of $\RabinOT{1/2}{k}$ (independently of $k$)...}  

\subsection{Lower Bounds for Equality Function}
\begin{corollary}\label{cor:eq}
Let a protocol having access to a $P_{UV}$ be an $\eps$-secure implementation of $\EQ{n}$ in the semi-honest model. Then %it holds for any $0<k\leq n$
\[\cHS{U}{V}\geq (1-\eps)k-(3\cdot 2^k-1)(\eps + h(\eps))-1\;,\]
and
\[\Hmin^{(3\cdot2^k+1)\eps}(U|V)\geq k-1\;,\]
for all $0<k\leq n$. If  $0\leq \eps \leq 1/(6\cdot2^k+2)$ and  $P_{UV}$ is equivalent to $m$ instances of $\OT{1}{2}{1}$, then
\cancel{
\[m \geq k-\log\left(\frac{1}{1-(3\cdot2^k+1)\eps}\right),\]
for all $0<k\leq n$. In particular, if $0\leq \eps \leq \frac{1}{2(3\cdot2^k+1)}$, then }
\begin{align*}
 m \geq k-1-(6\cdot 2^k+2)\eps\;,
\end{align*}
for all $0<k\leq n$.
\end{corollary}
\begin{proof}
We can restrict the input domains of both players to the same subsets of size $2^k$. Condition \eqref{nonTrivialF-1} will still be satisfied. Thus, the corollary follows immediately from Theorems~\ref{thm:H-imposs-n} and \ref{thm:Hmin-imposs-n}. 
%\qed
\end{proof}
There exists a secure reduction of $\EQ{n}$ to $\EQ{k}$ (\cite{BeiMal04}): Alice and Bob compare $k$ inner products of their inputs with random strings using $\EQ{k}$. This protocol is secure in the semi-honest model with an error of at most $2^{-k}$. Since there exists a circuit to implement $\EQ{k}$ with $k$ XOR and $k$ AND gates, it follows from \cite{GolVai87} that $\EQ{k}$ can be securely implemented using $k$ instances of $\OT{1}{4}{1}$ or $3k$ instances of $\OT{1}{2}{1}$ in the semi-honest model. Since $m$ instances of $\OT{1}{2}{1}$ are equivalent to a primitive $P_{UV}$ with $H(U|V)=m$, the bound of Corollary~\ref{cor:eq} is optimal up to a factor of~$3$. %\juerg{von mir aus koennte man die PR-boxen auch weglassen. oder vielleicht einfach das protocol unten angeben, und in einer fussnote sagen, dass man das auch ueber die kombination von \cite{WW05} und \cite{BBLMTU06} bekommt.} Since the non-local AND can be computed using two PR-Boxes \cite{BBLMTU06} and a PR-Box can be simulated from one instance 
%of $\OT{1}{2}{1}$ \cite{WW05}, one can actually show that 
We can improve the above construction with the following protocol that computes additive shares of $(x_1 \oplus y_1) \wedge (x_2 \oplus y_2)$ using two instances of $\OT{1}{2}{1}$: Alice chooses two random bits $r_1,r_2$ and inputs $r_1,r_1 \oplus x_1$ to the first and $r_2,r_2 \oplus x_2$ to the second instance. Bob uses $y_2$ as the choice bit for the first and $y_1$ as the choice bit for the second instance of OT. Bob receives two outputs $z_1=r_1 \oplus x_1y_2$ and $z_2=r_2 \oplus x_2y_1$. Setting $a=r_1\oplus r_2\oplus x_1x_2$ and $b=z_1\oplus z_2 \oplus y_1y_2$, we have $a \oplus b=x_1x_2\oplus y_1y_2\oplus x_1y_2\oplus x_2y_1=(x_1 \oplus y_1) \wedge (x_2 \oplus y_2)$. Thus, we can compute $\EQ{k}$ with $2(k-1)$ instances of $\OT{1}{2}{1}$. 

%\footnote{Using the fact that the non-local AND can be computed using two PR-Boxes \cite{BBLMTU06} and the simulation of a PR-Box using a $\OT{1}{2}{1}$ \cite{WW05} one can actually show that $\EQ{k}$ can be securely implemented using at most $2k$ instances of $\OT{1}{2}{1}$ (see Appendix~\ref{sec:and}) .}.

\subsection{Lower Bounds for Inner Product Function}
\begin{corollary}\label{cor:ip}
Let a protocol having access to a primitive $P_{UV}$ be an $\eps$-secure implementation of the inner product function $\IP{n}$ in the semi-honest model. Then it holds that
\begin{align*}
 \cHS{U}{V}&\geq n-1-4n(\eps+h(\eps))
\end{align*}
and
\begin{align*}
 \Chminteps{(3k+1)\eps}{U}{V}&\geq n-1\;.
\end{align*}
If  $P_{UV}$ is equivalent to $m$ instances of $\OT{1}{2}{1}$ and $0\leq \eps < 1/(6n+2)$, then
\cancel{
\[m \geq n-1-\log\left(\frac{1}{1-(3n+1)\eps}\right).\] 
In particular, if $0\leq \eps \leq \frac{1}{2(3n+1)\eps}$, then 
}
\begin{align*}
 m \geq n-1-(6n+2)\eps\;.
\end{align*}
\end{corollary}
\begin{proof}
Let $e_i \in \{0,1\}^n$ be the string that has a one at the $i$-th position and is zero otherwise. Let $\mS:=\{e_i:~1\leq i \leq n\}$. Then the protocol is an $\eps$-secure implementation of the restriction of the inner-product function to inputs from $\{0,1\}^n \times \mS$. Since this restricted function satisfies condition \eqref{nonTrivialF-1}, the statement follows from Theorem \ref{thm:H-imposs-n}.
%\qed
\end{proof}
If $\eps\leq 1/(8n)$, then it immediately follows from Corollary \ref{cor:ip} that we need at least $n-2$ calls to  $\OT{1}{2}{1}$ to compute $\IP{n}$ with an error of at most $\eps$. Consider the following protocol from \cite{BeiMal04} that is adapted to $\OT{1}{2}{1}$: Alice chooses $r=(r_1,\dots,r_{n-1})$ uniformly at random and sets $r_n:=\oplus_{i=1}^{n-1}r_i$. Then, for each $i$, Alice inputs $a_{i,0}:=r_i$ and $a_{i,1}:=x_i\oplus r_i$ to the OT and Bob inputs $y_i$. Bob receives $z_i$ from the OTs and outputs $\oplus _{i=1}^{n}z_i$. Since $\oplus _{i=1}^{n}z_i=\oplus_{i=1}^{n}(x_iy_i\oplus r_i)=(\oplus_{i=1}^{n}x_iy_i)\oplus(\oplus_{i=1}^{n} r_i)=\oplus_{i=1}^{n}x_iy_i=\IP{n}(x,y)$, the protocol is correct. The security for Alice follows from the fact that $z_1,\dots,z_n$ is a uniformly random string subject to $\oplus _{i=1}^{n}z_i=\IP{n}(x,y)$. Thus, there exists a perfectly secure protocol that computes $\IP{n}$ from $n$ instances of $\OT{1}{2}{1}$. Hence, Corollary~\ref{cor:ip} is almost tight.

\subsection{Lower Bounds for Protocols implementing OLFE}
We will now show that Theorems \ref{thm:H-imposs-n} and \ref{thm:I-imposs-n} also imply bounds for \emph{oblivious linear function evaluation} ($\OLFE{q}$), which is defined as follows:
\begin{itemize}
\item 
For any finite field $GF(q)$ of size $q$, $\OLFE{q}$ is the primitive where Alice has an input $a,b \in GF(q)$ and Bob has an input $c \in GF(q)$. Bob receives $d = a + b \cdot c \in GF(q)$.
\end{itemize}
%Our lower bound is a simple consequence of the fact that $\OLFE{q}$ can be used to implement $\OT{1}{2}{\log(q)}$.
\begin{corollary} \label{cor:olfe1}
Let a protocol having access to $P_{UV}$ be an $\eps$-secure implementation of $m$ instances of
$\OLFE{q}$ in the semi-honest model. Then
\begin{align}
\cHS{U}{V} & \geq m \log q - 5 (\eps m \log q + h(\eps))\;,\label{olfe:first}\\
\cHS{V}{U} & \geq m \log q - 5 (\eps m \log q + h(\eps))\;,\label{olfe:second}\\
\cIS{U}{V}{\C{U}{V}} & \geq m \log q - 7(\eps m \log q + h(\eps))\label{olfe:third}\;.
\end{align}
\end{corollary}

\begin{proof}
Inequalities~\eqref{olfe:first} and~\eqref{olfe:third} follow from Theorem~\ref{thm:H-imposs-n} and Theorem~\ref{thm:I-imposs-n}. Furthermore, it has been shown in \cite{WolWul06} that $\OLFE{q}$ is symmetric. Hence, a violation of~\eqref{olfe:second} would imply a violation of the lower bound in~\eqref{olfe:first}.
%\qed
\end{proof}

\subsection{Lower Bounds for OT in the Malicious Model}\label{subsec:sec:mal:ot}

In Appendix~\ref{appendix:sec:mal}, we show that lower bounds in the semi-honest model imply almost the same bounds in the malicious model. In the following, we generalize these results by allowing a dishonest Bob to additionally receive randomness $V'$.
\cancel{Next, we show that the lower bound for $\OT{1}{2}{n}$ implied by Theorem~\ref{thm:Hmin-imposs-n} also holds in the malicious model, where a dishonest Bob additionally receives randomness $V'$. This generalizes the result of Appendix~\ref{appendix:sec:mal}, where a dishonest Bob only receives $V$.}
 Moreover, the following provides a stronger impossibility result, in the case when $V'$ is trivial, than the one that follows from the combination of Lemma~\ref{lem:malicious:ot} and Theorem~\ref{thm:Hmin-imposs-n}.
\begin{corollary}\label{cor:mal-bound}
 Let a protocol be an $\eps$-secure implementation of $\OT{1}{2}{n}$ in the malicious model from randomness $(U,VV')$. Then
\begin{align*}
 \Chminteps{7\eps}{U}{VV'} \geq  k\;.
\end{align*}
\end{corollary}
\begin{proof}
We consider only honest players, but allow the simulator to change the inputs to the ideal OT and the outputs from the ideal OT. Lemma~\ref{lem:max-entropyUpperBound-n} holds in the weak semi-honest model and, therefore, also in the malicious model. Thus, we have $\Chmaxeps{6\eps}{X}{UM,C=c} =0$, where $C$ is the choice bit of Bob. The security of the protocol implies that there exists a randomized function $S_B$ such that 
\begin{align}\label{ineq:sim_B}
 \dis(P_{XS_B(C,X_{\tilde C})},P_{XVV'M})\leq \eps, 
\end{align}
where $\tilde C$ is the input to the ideal OT by the simulator. Therefore, we get
\begin{align*}
 \Chminteps{\eps}{X}{VV'M,C=c} &\geq \Chminteps{}{X}{S_B(c,X_{\tilde C})}\nonumber\\
&\geq \Chminteps{}{X}{X_{\tilde C}}\nonumber\\
&\geq k\;.
\end{align*}
As in the proof of Theorem~\ref{thm:Hmin-imposs-n}, this implies
\begin{align*}
k\leq\Chminteps{\eps}{X}{VV'M,C=c}&\leq \Chminteps{7\eps}{U}{VV'}\;.
\end{align*}
%\qed
\end{proof}

In the same way, we can show that the impossibility result for implementations of $\OT{1}{2}{k}$ that follows from Theorem~\ref{thm:H-imposs-n} also holds in the malicious model.
\begin{corollary} \label{cor:H-imposs-n-mal}
Let a protocol be an $\eps$-secure implementation of $\OT{1}{2}{k}$ in the malicious model from randomness $(U,VV')$. Then
\begin{align*}
\cHS{U}{VV'}\geq k-6(k\eps+h(\eps))\;. 
\end{align*}
\end{corollary}
\begin{proof}
Since Lemma~\ref{lem:entropyUpperBound-n} also holds in the malicious model, inequality \eqref{eq:mono} implies that
\begin{align*}
 \cHS{X}{ UVV'M,C=c}&\leq \cHS{X}{ UM,C=c}\\
&\leq 4(k\eps+ h(\eps))\;.
\end{align*}
We can use inequalities~\eqref{lem:h-smooth} and \eqref{ineq:sim_B} to obtain
\begin{align*}
 \cHS{X}{VV'M,C=c} &\geq \cHS{X}{S_B(c,X_{\tilde C})}-\eps\cdot  2k-h(\eps)\nonumber\\
&\geq \cHS{X}{X_{\tilde C}}-\eps \cdot 2 k-h(\eps)\nonumber\\
&\geq k-\eps \cdot 2 k-h(\eps)\;.
\end{align*}
As in the proof of Theorem~\ref{thm:H-imposs-n}, this implies
\begin{align*}
\cHS{U}{VV'}\geq k-6(k\eps+h(\eps)). 
\end{align*}
\end{proof}

\cancel{
As an example, we consider the following reduction where the bound of Corollary \ref{cor:mal-bound} can be applied: Let $(U,V)=((X_0,X_1),(C,X_C))$ be random variables corresponding to a randomized $\OT{1}{2}{k}$ and let $V'$ be defined by a conditional distribution $P_{V'|X_0X_1C}=P_{V'|X_C}$ such that $\Chminteps{}{X_{1-C}}{VV'}\geq \alpha k$, where $0<\alpha\leq 1$. We consider the protocol where Alice chooses $f_0,f_1$ from a family of two-universal hash functions from $\{0,1\}^k$ to  $\{0,1\}^{k'}$, outputs $\tilde x_0:=f_0(x_0)$ and $\tilde x_1:=f_0(x_0)$ and sends $f_0,f_1$ to Bob who outputs $(C,f_C(x_C))$. Privacy Amplification \cite{BeBrRo88,ILL89,DodisORS08} implies that this protocol implements a randomized OT with an error of at most $\eps$ if $k' \leq \alpha k-2\log(1/\eps)$. Thus, one can also implement a $\OT{1}{2}{k'}$.
}

Corollary~\ref{cor:mal-bound} can be applied to  implementations of $\OT{1}{2}{k}$ from Universal OT over bits. Universal OT \cite{BrCrWo03} is a weakened variant of Bit-OT where a dishonest Bob can choose a channel $P_{Y|X}$ such that $H(X_0X_1|Y)\geq \alpha$, where $(X_0,X_1) \in \sbin \times \sbin$ are uniform and $Y$ is the output of the channel $P_{Y|X}$, and learns $Y$. One choice of a dishonest Bob is the channel that outputs both inputs with probability $1-\alpha$ and one of the inputs, $X_c$, with probability $\alpha$. This primitive can be implemented from randomness $(U,VV')=((X_0,X_1),(C,X_C,V'))$, where $(U,V)$ corresponds to a randomized $\OT{1}{2}{1}$ and $V'=X_{1-C}$ with probability $1-\alpha$ and $V'=\bot$ otherwise. For this primitve we have $\cHS{U}{VV'}\leq \alpha$ and, therefore, for $n$ independent instances $\cHS{U^n}{(VV')^n} \leq \alpha n$. Thus, Corollary~\ref{cor:mal-bound} implies that $k \leq \alpha n +6(k\eps+h(\eps))$. As Univeral OT is strictly weaker than this 
primitive, the same bound also applies to Universal OT. The protocol proposed in \cite{CreSav06} which implements $\OT{1}{2}{k}$ from $n$ instances of Universal OT asymptotically achieves a rate $k/n$ of $\alpha$ \cite{Win12}. Our lower bound now shows that this is in fact optimal.

\cancel{
An instance of $\alpha n$-UniversalOT can be implemented from distributed randomness $(U,VV')$ such that $U=(X_0,X_1)$ and $(V,V')=((C,X_C),Y)$ where $(X_0,X_1,C,X_C)$ are random variables with the distribution of $n$ randomized Bit-OTs. %$V'$ is  equal to $X_{1-C}$ with probability $\frac{1}{2^{\alpha n}}$ and a uniform string from the remaining strings $(\{0,1\}^n\setminus X_{1-C})$, otherwise. 
$V'$ is equal to the first $n-\lceil \alpha n \rceil$ bits of $X_{1-C}$. Then $\chmin{}{U}{VV'}=\lceil \alpha n \rceil\geq \alpha n$

\begin{lemma}\label{lem:universal-bound}
Let $0<\alpha\leq 1$. If there exists an $\eps$-secure implementation of $\OT{1}{2}{\ell}$ from an $\alpha n$-UniversalOT, where $\eps \leq 1/14$, then
\begin{align*}
\alpha n\geq \ell-1-14\eps\;.
%\alpha n\geq \ell-1-\log\left(\frac{1}{1-7\eps}\right).
\end{align*}
\end{lemma}
\begin{proof}
Corollary \ref{cor:mal-bound} implies that 
\begin{align*}
 \Hmin^{7\eps}(U\mid VV') \geq \ell\;.
\end{align*}
The statement then follows from
\begin{align*}
\Chminteps{7\eps}{U}{VV'}&\leq -\log(2^{-\lceil \alpha n \rceil}-7\eps/2^{\lceil \alpha n \rceil})\\
  %&\leq -\log(2^{-\lceil\alpha n\rceil}-7\eps/2^{\lceil \alpha n \rceil})\\
  &=-\log(2^{-\lceil\alpha n\rceil}(1-7\eps))\\
  &=\lceil\alpha n\rceil  -\log(1-7\eps)\\
  &\leq \alpha n + 1 -\log(1-7\eps)\;.
\end{align*}
%The statement follows.
%\qed
\end{proof}
}

\input{reversingStringOT_peerreview}
\input{quantum_peerreview}

\section*{Acknowledgment}
The authors would like to thank Renato Renner and Marco Tomamichel for helpful discussions and the referees for their useful comments. 

\appendices
\input{malicious_peerreview}
\input{smooth_entropies_ieee_peerreview}

\input{Lemmas_peerreview}
% Generated by IEEEtran.bst, version: 1.13 (2008/09/30)

\end{document}

%% file: reversingStringOT_peerreview.tex
\section{Quantum Reductions: Reversing String OT Efficiently}\label{sec:quantumReductions}
As the bounds of the last section generalize the known bounds for perfect implementations of OT from \cite{Beaver96,DodMic99,WolWul05,WolWul08} to the statistical case, it is natural to ask whether similar bounds also hold for quantum protocols, i.e., if the bounds presented in \cite{SaScSo09} can be generalized to the statistical case. We give a negative answer to this question by presenting a statistically secure quantum protocol that violates these bounds. Thereto we introduce the following functionality\footnote{In this section we will use the notion commonly used in the UC framework, that is slightly different from the rest of the paper.} $\mF_{\texttt{MCOM}}^{A \rightarrow B,k}$ that can be implemented from $\mF_{\texttt{OT}}^{A \rightarrow B,k}$ (i.e., $\OT{1}{n}{k}$)  as we will show.
 
\begin{definition}[Multi-Commitment]
The functionality $\mF_{\texttt{MCOM}}^{A \rightarrow B,k}$ behaves as follows: Upon (the first) input (\texttt{commit}, b) with $b \in \{0, 1\}^{k}$ from Alice, send \texttt{committed} to Bob. Upon input \texttt{(open,T)} with $T \subseteq [k]$ from Alice send (\texttt{open}, $b_T$) to Bob. All communication/input/output is classical. We call Alice the sender and Bob the receiver. 
\end{definition}

An instance of $\OT{1}{2}{k}$ can be implemented from $m = O(k + \kappa)$ bit commitments with an error of $2^{-\Omega(\kappa)}$ \cite{BBCS92,Yao95,DFLSS09}. In the protocol, Alice sends $m$ BB84-states to Bob who measures them either in the computational or in the Hadamard basis. To ensure that he really measures Bob has to commit to the basis he has measured in and the measurement outcome for every qubit received. Alice then asks Bob to open a small subset $\mT$ %of size $\alpha m$
of these pairs of commitments. OT can then be implemented using further classical processing (see Section \ref{sec:OTtoSC} for a complete description of the protocol). This protocol implements oblivious transfer that is statistically secure in the quantum \emph{universal composability} model \cite{Unruh10}.
Obviously the construction remains secure if we replace the commitment scheme with $\mF_{\texttt{MCOM}}^{A \rightarrow B,2m}$.

Next, we show that $\mF_{\texttt{MCOM}}^{A \rightarrow B,k}$ can be implemented from the oblivious transfer functionality $\mF_{\texttt{OT}}^{A \rightarrow B,k}$ (see \cite{Unruh10} for a definition of  $\mF_{\texttt{OT}}^{A \rightarrow B,k}$) using Protocol MCOMfromOT.

As it is done in the proofs of \cite{Unruh10}, we assume that all communication between the players is over secure channels and we only consider static adversaries.
\begin{lemma}\label{lem:UC-sec}
Protocol MCOMfromOT  is statistically secure and universally composable and realizes $\mF_{\texttt{MCOM}}^{A \rightarrow B,k}$ with an error of $2^{-\kappa/2}$ using $\kappa$ instances of $\mF_{\texttt{OT}}^{A \rightarrow B,k}$.
\end{lemma}

\begin{figure}[!t]
\begin{framed}
\noindent Protocol \textbf{MCOMfromOT \\}
Inputs: Alice has an input $b=(b_1,\dots,b_k) \in \{0,1\}^{k}$ in \texttt{Commit}. Bob has an input $T\subseteq [k]$ in \texttt{Open}.\\
\texttt{Commit($b$)}:\\
For all $1 \leq i \leq \kappa$: 
\begin{enumerate}
 \item Alice and Bob invoke $\mF_{\texttt{OT}}^{A \rightarrow B,k}$ with random inputs $x^i_0,x^i_1 \in \{0,1\}^{k}$ and $c^i \in \{0,1\}$.
 %\item Alice gets random $x^k_0,x^k_1 \in \{0,1\}^l$ and Bob gets $y^k:=x^k_{c^k}$.
 \item Bob receives $y^i:=x^i_{c^i}$ from  $\mF_{\texttt{OT}}^{A \rightarrow B,k}$.
 \item Alice sends $m^i:=x_0^i\oplus x_1^i\oplus b$ to Bob.
\end{enumerate}
\texttt{Open}(T):
\begin{enumerate}
 \item Alice sends $b_T$, $T$ and $(x^i_{0})_T,(x^i_{1})_T$ for all $1 \leq i \leq \kappa$ to Bob.
 \item If $(m^i)_T=(x_{0}^i\oplus x_{1}^i \oplus b^i)_T$ and $(y^i)_T=(x^i_{c})_T$ for all $1 \leq i \leq \kappa$, Bob accepts and outputs $b_T$, otherwise he rejects.
\end{enumerate}
\end{framed}
\end{figure}
%\{4mm}

\begin{proof}
%Note that we assume that all communication between the players is over secure channels and we only consider static adversaries. 
The statement is obviously true in the case of no corrupted parties and in the case when both the sender and the receiver are corrupted. We construct for any adversary $\mA$ a simulator $\mS$ that runs a copy of $\mA$ as a black-box. In the case where the sender is corrupted, the simulator $\mS$ can extract the commitment $b$ from the input to $\mF_{\texttt{OT}}^{A \rightarrow B,k}$ and the messages except with probability $2^{-\kappa/2}$ as follows: We define the extracted commitment as $b_i:=\maj(m^1_i \oplus x_{0,i}^1 \oplus x_{1,i}^1,\dots,m^\kappa_i \oplus x_{0,i}^\kappa \oplus x_{1,i}^\kappa)$ for all $1 \leq i \leq k$ where $\maj$ denotes the majority function. Let $\mT$ be a (non-empty) subset of $[k]$ and let $\tilde b\in \{0,1\}^{k}$ such that $\tilde b_{\mT} \neq b_{\mT}$. An honest receiver accepts $\tilde b_{\mT}$ together with $\mT$ in \texttt{Open} with probability at most $2^{-\kappa/2}$ as follows: There must exist $j \in \mT$ such that $b_j \neq \tilde b_j$. Then the sender needs to change either $x^i_{0,j}$ or $x^i_{1,j}$ for at least $\kappa/2$ instances $i$. Thus, the simulator extracts the bit $b$ in the commit phase as specified before and gives $(\texttt{commit},b)$ to $\mF_{\texttt{MCOM}}^{A \rightarrow B,k}$. Upon getting $(\tilde b_{\mT},\mT)$ from the adversary, the simulator gives $(\texttt{open},\mT)$ to $\mF_{\texttt{MCOM}}^{A \rightarrow B,k}$, if $\tilde b_\mT = b_\mT$, otherwise it stops. Therefore, any environment can distinguish the simulation and the real execution with an advantage of at most $2^{-\kappa/2}$. In the case where the receiver is corrupted, the simulator $\mS$, upon getting the message \texttt{committed} from $\mF_{\texttt{MCOM}}^{A \rightarrow B,k}$ and the choice bit $c^i$, chooses the output $y^i$ from  $\mF_{\texttt{OT}}^{A \rightarrow B,k}$ and the message $m^i$ uniformly and independently at random for all $i$. In the open phase the simulator $\mS$ gets $(\mT,b_{\mT})$ and simulates the messages of an honest sender by setting $(x^i_{1-c^i})_\mT:= (m^i)_\mT \oplus (y^i)_\mT \oplus b_\mT$ and $(x^i_{c^i})_\mT:=(y^i)_\mT$ for all $i$. This simulation is perfectly indistinguishable from the real execution.
%\qed
\end{proof}
Any protocol that is statistically secure in the classical universal composability model \cite{Canetti01} is also secure in the quantum universal composability model \cite{Unruh10}. Together with the proofs from \cite{DFLSS09,Unruh10}, we, therefore, obtain the following theorem.
\begin{theorem}\label{thm:revers-string-ot}
There exists a protocol that implements $\OT{1}{2}{k'}$ with an error $\eps$ from $\kappa = O(\log 1/\eps)$ instances of $\OT{1}{2}{k}$ in the opposite direction where $k' = \Omega(k)$ if $k=\Omega(\kappa)$.
\end{theorem}
Since we can choose $k \gg \kappa$, this immediately implies that the bound of Corollary~\ref{cor:H-imposs-n} does not hold for quantum protocols. Similar violations can be shown for the other two lower bounds (given in Corrollary~\ref{cor:main2}). For example, statistically secure and universally composable\footnote{Stand-alone statistically secure commitments based on stateless two-party primitives are universally composable~\cite{DGMN08}.} commitments can be implemented from shared randomness $P_{UV}$ that is distributed according to $\RabinOT{p}{}$ at a rate of $\Hshannon{U}{V}=1-p$ \cite{WiNaIm03}. Using Theorem~\ref{thm:optimal-protocol}, one can implement $\mF_{\texttt{OT}}^{B \rightarrow A,k}$ with $k \in \Omega(n(1-p))$ from $n$ copies of $P_{UV}$. Since $\Ishannon{U}{V}=p$, quantum protocols can also violate the bound of Corollary \ref{cor:I-imposs-n}.

It has been an open question whether noiseless quantum communication can increase the commitment capacity~\cite{WiNaIm03}. Our example implies a positive answer to this question.

%% file: quantum_peerreview.tex
\section{Impossibility Results for Quantum Oblivious Transfer Reductions}
%\severin{computational und hadamard basis irgendwo einfuehren}
We consider finite-dimensional Hilbert spaces $\hi{}$. A quantum state $\rho$ is a positive semi-definite operator on $\hi{}$ satisfying $\tr(\rho) = 1$. We use the notation $\rho_{AB}$ for a state on $\hi{A}\kron\hi{B}$ and define the marginal state $\rho_A:=\ptr{B}{\rho_{AB}}$.  We use the symbol $\idA$ to denote either the identity operator on $\hi{A}$ or the identity operator on the states on $\hi{A}$; it should be clear from the context which one is meant. Given a finite set $\mX$ and an orthonormal basis $\{\ket{x} \mid x\in\mX\}$  of a Hilbert space $\hi{$\mX$}$ we can encode a classical probability distribution $P_X$ as a quantum state $\rho_X=\sum_{x\in\mX}P_X (x) \vecstate{x}$. We define the state corresponding to the uniform distribution on $\mX$ as $\tau_{\mX}:=\frac{1}{|\mX|}\sum_{x\in\mX}\vecstate{x}$.
A state $\rho_{XB}$ on $\hi{$\mX$}\kron\hi{B}$ is a classical-quantum or cq-state if it is of the form $\rho_{XB}=\sum_{x \in\mX}p_x\vecstate{x}\kron\rho_B^x$. The \emph{Hadamard transform} is the unitary described by the matrix $H=\tfrac{1}{\sqrt{2}}\left(\begin{smallmatrix} 1&1\\ 1&-1 \end{smallmatrix}\right)$ in the computational basis $\{\ket{0},\ket{1}\}$. For $x,\theta \in \sbin^n$, we write $H^\theta \ket{x}$ for the state $H^\theta \ket{x}=H^{\theta_1}\ket{x_1}\ldots H^{\theta_n}\ket{x_n}$. We also call states of this form \emph{BB84-states}. When speaking of the basis $\theta \in \sbin^n$ we mean the basis $\{H^\theta\ket{x} \mid x\in \sbin^n\}$. For a given basis $\{\ket{x_1},\ldots,\ket{x_d}\}$ we say that we \emph{measure in basis $\mB$} to indicate that we perform a projective measurement given by the operators $P_k=\proj{x_k}$ for all $k\in[d]$.  
 %We denote the computational basis $\{\ket{0},\ket{1}\}$ of the qubit Hilbert space $\mathbb{C}^2$ by $'+'$ and the Hadamard basis $\{H\ket{0},H\ket{1}\}$, where $H$ is the Hadamard transform $H=\frac{1}{\sqrt{2}}\left(\begin{smallmatrix}1 & 1\\1&-1\end{smallmatrix}\right)$, by $'\times'$.
The \emph{trace distance} between two quantum states $\rho$ and $\tau$ is defined as
\begin{align*}
 \dis(\rho,\tau):=\max_{\mE}D(\mE(\rho),\mE(\tau))\;.%\frac{1}{2} \norm{\rho - \tau}{1},%+\frac{1}{2}|\tr\,\rho-\tr\, \tau|
\end{align*}
%where the Schatten 1-norm is defined as $\norm{X}{1}=\tr |X|=\tr \sqrt{X^\dagger X}$. %A quantum state can be used to encode classical probability distributions $P_X$ on a finite set $\mX$, i.e., 
where the maximum is over all POVMs and $\mE(\rho)$ is the probability distribution of the measurement outcomes. In particular, for any two cq-states $\rho_{X A}$ and $\sigma_{X A}$, $\dis(\rho_{X A}, \sigma_{XA}) \leq \eps$
implies that for any measurement $G$ on system $A$, we have
\begin{align}\label{eq:guess0}
 \big | \Pr[ G(\rho_A) = X] - \Pr[ G(\sigma_A) = X] \big | \leq \eps\;.
\end{align}
If we choose $\sigma_{X A} := \tau_X \otimes \sigma_A$, this implies that 
\begin{equation} \label{eq:guess}
 \Pr[ G(\rho_A) = X] \leq \frac12 + \eps\;.
\end{equation}

The \emph{conditional von Neumann entropy} is defined as
\[ \chvn{A}{B}{\rho} := \hvn{\rho_{AB}}{} - \hvn{\rho_B}{}\;,\]
where $H(\rho) := \tr( - \rho \log(\rho))$. The Alicki-Fannes inequality \cite{AliFan04} implies that 
\begin{equation} \label{eq:alifan}
 \chvn{A}{B}{\rho} \geq (1-4\eps) \cdot \log |A| - 2h(\eps)\;,
\end{equation}
for any state $\rho_{AB}$ with  $\dis(\rho_{AB}, \tau_A \otimes \rho_{B}) \leq \eps$.
Let $\rho_{XB}$ be a state that is classical on $X$. If there exists a measurement on $B$ with outcome $X'$ such that $\Pr[X' \neq X] \leq \eps$, then
\begin{equation} \label{eq:fano-q}
  \chvn{ X}{B}{\rho} \leq \chvn{X}{X'}{} \leq h(\eps) + \eps \cdot \log |X|\;.
\end{equation}
Let $\rhoABC$ be a tripartite state. Subadditivity and the triangle inequality \cite{AL70} imply that 
\begin{align}\label{eq:cond-mi}
 %\chvn{A}{B}{\rho}-\chvn{A}{BC}{\rho}\leq 2\hvn{C}{\rho}.
\chvn{A}{BC}{\rho}\geq\chvn{A}{B}{\rho}-2\hvn{C}{\rho}\;.
\end{align}
The conditional entropy $\chvn{A}{B}{\rho}$ can decrease by at most $\log|Z|$ when conditioning on an additional classical system $Z$, i.e., for any tripartite state $\rhoABZ$ that is classical on $Z$ with respect to some orthonormal basis $\{\ket{z}\}_{z \in \mZ}$, we have
\begin{align}\label{chain-rule-c}
 \chvn{A}{BZ}{\rho}\geq \chvn{A}{B}{\rho}-\log |Z|\;.
\end{align}
The next lemma can be obtained by applying the asymptotic equipartition property to the corresponding lemma for the smoothed min-entropy in \cite{WTHR11}. It shows that the entropy $\chvn{A}{BC}{\rho}$ cannot increase too much when a projective measurement is applied to system $C$.

\begin{lemma}
  \label{lem:data-processing-bound}
  Let $\rhoABC$ be a tri-partite state. Furthermore, let $\cM$ be a projective measurement in the basis $\{\ket{z} \}_{z \in \mZ}$ on C and 
  $\idx{\rho}{ABZ} := (\id_{AB} \kron \cM )(\rhoABC)$. % where $\idx{\cI}{AB}$ is the identity operation on A and B. 
Then,
  \begin{align*}
  	\chvn{A}{BC}{\rho} \geq \chvn{A}{BZ}{\rho} - \log \abs{Z} \; .
  \end{align*}
\end{lemma}

\subsection{Security Definition}

A protocol is an $\eps$-secure implementation of OT in the malicious model if for any adversary $\mA$ attacking the protocol (real setting), there exists a \emph{simulator} $\mS$ using the ideal OT (ideal setting) such that  for all inputs of the honest players the real and the ideal setting can be distinguished with an advantage of at most $\eps$. This definition implies the following three conditions (see also \cite{FS09}):
\begin{itemize}
 \item Correctness: If both players are honest, Alice has random inputs $(X_0,X_1)\in \sbin^k\times\sbin^k$ and Bob has input $c\in\sbin$, then Bob always receives $X_c$ in the ideal setting. This implies that in an $\eps$-secure protocol, Bob must output a value~$Y$, where
\begin{equation} \label{eq:corr}
 \Pr[Y \neq X_c] \leq \eps\;.
\end{equation}

 \item Security for Alice: Let Alice be honest and Bob malicious, and let Alice's input be chosen uniformly at random. In the ideal setting, the simulator must provide the ideal OT with a classical input $C' \in \{0,1\}$. He receives the output $Y$ and then outputs a quantum state $\sigma_B$ that may depend on $C'$ and $Y$. The output of the simulator together with classical values $X_0$, $X_1$ and $C'$ now defines the state $\sigma_{X_0 X_1 B C'}$. Since $X_{1-C'}$ is random and independent of $C'$ and $Y$, we must have
\begin{align} \label{eq:secA1}
\sigma_{X_{1-C'} X_{C'} B C'}=\pi_{\sbin^k} \otimes \sigma_{X_{C'} B C'}
\end{align}
and
\begin{align} \label{eq:secA2}
\dis(\sigma_{X_0 X_1 B},\rho_{X_0 X_1 B}) \leq \eps\;,
\end{align}

where $\rho_{X_0 X_1 B}$ is the resulting state of the protocol.\footnote{The standard security definition of OT considered here requires Bob's choice bit to be fixed at the end of the protocol. To show that a protocol is insecure, it suffices, therefore, to show that Bob can still choose after the termination of the protocol whether he wants to receive $x_0$ or $x_1$. Lo in \cite{Lo97} shows impossibility of OT in a stronger sense, namely that Bob can learn all of Alice's inputs.}
\item
Security for Bob: If Bob is honest and Alice malicious, the simulator outputs a quantum state $\sigma_A$ that is independent of Bob's input $c$. Let $\rho_{A}^c$ be the state that Alice has at the end of the protocol if Bob's input is $c$. The security definition now requires that $\dis(\sigma_A,\rho_{A}^c)\leq \eps$ for $c \in \{0,1\}$. By the triangle inequality, we get
\begin{equation} \label{eq:secB}
\dis(\rho_{A}^0, \rho_A^1) \leq 2\eps\;.
\end{equation}
\end{itemize}

Note that the Conditions \eqref{eq:corr} - \eqref{eq:secB} are only necessary for the security of a protocol, they do \emph{not} imply that a protocol is secure. 

\cancel{
\begin{lemma}\label{lem:uhlmann}
Let $\ket{\psiAB^0}$ and $\ket{\psiAB^1}$ be states with $\dis(\rhoB^0,\rhoB^1)\leq \eps$ where $\rhoB^x=\ptr{A}{\ket{\psiAB^x}\bra{\psiAB^x}}$. Then there exists a unitary $U_{A}$ such that 
\[\dis(\ket{\phiAB^1}\bra{\phiAB^1},\ket{\psiAB^1}\bra{\psiAB^1})\leq \sqrt{2\eps}\]
with $\phiAB^1=(U_{A}\kron \idB)\ket{\psiAB^0}$.
\end{lemma}
\begin{proof}
$\dis(\rhoB^0,\rhoB^1)\leq \eps$ implies $F(\rhoB^0,\rhoB^1)\geq 1-\eps$. From Uhlmann's theorem we know that there exists a unitary $U_A$ such that
\[F(\ket{\phiAB^1}\bra{\phiAB^1},\ket{\psiAB^1}\bra{\psiAB^1})\geq 1-\eps\] where $\ket{\phiAB^1}=(U_{A}\kron \idB)\ket{\psiAB^0}$. Since $\dis(\rho,\tau)\leq \sqrt{1-F(\rho,\tau)^2}$ for any $\rho,\tau \in \normstates{\h}$ \cite{FG99}, we have $\sqrt{1-\dis(\ket{\phiAB^1}\bra{\phiAB^1},\ket{\psiAB^1}\bra{\psiAB^1})^2}\geq 1-\eps$. Hence,
\begin{align*}
\dis(\ket{\phiAB^1}\bra{\phiAB^1},\ket{\psiAB^1}\bra{\psiAB^1})\leq \sqrt{1-(1-\eps)^2}\leq \sqrt{2\eps}
\end{align*}
\end{proof}
}
In the following we present two impossibility results for quantum protocols that implement $\OT{1}{2}{k}$ using a bit commitment functionality or randomness distributed to the players. We consider protocols which are information-theoretically secure. In particular, we assume that the adversary has unlimited memory space and can apply arbitrary quantum operations to his whole quantum system. Our proofs use similar techniques as the impossibility results in \cite{Mayers97,LoChau97,Lo97}. %\juerg{Hier ist komisch dass es kein "second" gibt. ich finde wir koennten das einfach weglassen.} First, the protocol is replaced by a purified version of the protocol that is equivalent in a certain sense. In particular the purified version has the same security properties as the original protocol and the impossibility of the former implies the impossibility of the latter.

We assume that the two parties, Alice  and Bob, have access to a noiseless quantum and a noiseless classical channel. The protocol proceeds in rounds, where in any round of the protocol, the parties may perform an arbitrary quantum operation on the system in their possession. This operation can be conditioned on the available classical information and generates the inputs to the communication channels. The quantum channel transfers a part of one party's system to the other party. The classical channel measures the input in a canonical basis and sends the outcome of the measurement to the receiver. We assume that the total number of rounds of the protocol is bounded by a finite number. Since we can always introduce empty rounds, this corresponds to the assumption that the number of rounds is equal in every execution of the protocol. 

All quantum operations of both parties can be purified by introducing an additional memory space: Any quantum operation $\mE$ can be simulated by adding an ancillary system, applying a unitary on the composite system, and then tracing out part of the remaining system. More precisely, for any TP-CPM~$\mE$ from $\normstates{\hA}$ to $\normstates{\hB}$, there exists a Hilbert space $\hR$, a unitary $U$ acting on $\h_{ABR}$ and a pure state $\sigma_{BR}\in \normstates{\h_{BR}}$  such that
\begin{align}\label{eq:stinespring}
 \mE(\rhoA)=\ptrace{AR}{U (\rhoA \otimes \sigma_{BR}) U^\dagger}.
\end{align}
This is known as the \emph{Stinespring dilation}~\cite{Stine55} of $\mE$. Thus, we can assume that the parties apply in every round of the protocol a unitary to their system conditioned on the information shared over the classical channel. In particular, we can assume that the system remains in a pure state conditioned on the information shared over the classical channel if the initial state of the protocol is pure. Since a malicious player can purify all his quantum operations in the original protocol without being detected, the purified protocol is secure according to our definition if the original protocol is secure.

An important tool in our impossibility proofs is the following technical lemma from \cite{WTHR11}, which generalizes a result already used in \cite{Mayers97,LoChau97,Lo97}.
\begin{lemma}\label{lem:classical-attack}
For $b \in \{0,1\}$, let 
\[\rho_{XX'AB}^b=\sum_xP_b(x)\vecstate{x}_{X}\kron \vecstate{x}_{X'}\kron \vecstate{\psiAB^{x,b}}\]
with $\dis(\rho_{X'B}^0,\rho_{X'B}^1)\leq \eps$. Then there exists a unitary $U_{AX}$ such that 
\[\dis(\rho_{XX'AB}'^1,\rho_{XX'AB}^1)\leq\sqrt{2 \eps} \]
where $\rho_{XX'AB}'^1=(U_{XA}\kron \id_{X'B})\rho_{XX'AB}^0(U_{XA}\kron \id_{X'B})^\dagger$.
\end{lemma}
\cancel{
\begin{proof}
Define $\ket{\psi_{XX'X''AB}^b}:=\sum_x\sqrt{P_b(x)}\ket{x}_X\kron\ket{x}_{X'}\kron\ket{x}_{X''}\kron\ket{\psiAB^{x,b}}$ and let
\[\rho_{X'X''B}^b=\ptrace{XA}{\vecstate{\psi_{XX'X''AB}^b}}.\]%\rhob_{ABC_AC_B}}=\rho_{BC_B}^b.\]
Then
\[\dis(\rho_{X'X''B}^0,\rho_{X'X''B}^1)=\dis(\rho_{X'B}^0,\rho_{X'B}^1)\leq \eps\]
Thus, Lemma \ref{lem:uhlmann} implies the existence of a unitary $U_{XA}$ such that 
\[\dis(\vecstate{\phi_{XX'X''AB}^1},\vecstate{\psi_{XX'X''AB}^1})\leq \sqrt{2\eps}\]
with $\ket{\phi_{XX'X''AB}^1}=(U_{XA}\kron \id_{X'X''B})\ket{\psi_{XX'X''AB}^0}$. The statement then follows from the fact that taking the partial trace over $X''$ cannot increase the trace distance and commutes with the unitary $U_{XA}$.
\end{proof}
}

First, we consider protocols where the players can use a certain number $n$ of ideal bit commitments as a resource to implement an oblivious transfer.
%Let Alice choose her inputs $X_0$ and $X_1$ uniformly at random and let Bob's input be $c$.

\begin{theorem} \label{thm:imposs:com}
Any protocol that implements a $\OT{1}{2}{k}$ with an error of at most~$\eps$ , where $0\leq \eps \leq 0.002$, from black-box bit commitments, has to use at least $(1- 3 \sqrt{\eps}) \cdot k - 3 h(\sqrt{\eps})$
bit commitments.
\end{theorem}

\begin{proof}
Let $n_A$ be the number of bit commitments from Alice to Bob and $n_B$ the number of bit commitments from Bob to Alice used in the protocol and $n=n_A+n_B$. Let Alice choose her inputs $X_0$ and $X_1$ uniformly at random. Let the final state of the protocol on Alice's and Bob's system be $\rho_{AB}^c$, when both players are honest and Bob has input $c \in \{0,1\}$. If Bob is executing the protocol honestly using input $c=1$, he can compute $X_1$ with an error of at most $1-\eps$. Since the protocol is $\eps$-secure for Alice, we can conclude from Lemma~\ref{lem:semihonest} that $\dis(\rho_{X_0 B}^1, \tau_{X_0} \otimes \rho_{ B}^1) \leq 5 \eps\;.$ Equation~\eqref{eq:alifan} implies that 
\begin{align}
 \chvn{X_0}{B}{\rho^1} %&\geq (1- 20 \eps) \cdot k - 2h(5 \eps)\nonumber\\
 &\geq (1- 20 \eps) \cdot k - 10 h(\eps)\;.
\end{align}
In the following, we consider a modified protocol that does not use the bit commitment functionality and is not necessarily secure for Alice. In this protocol we make Bob more powerful in the sense that he can simulate the original protocol locally. Thus, the modified protocol is still secure for Bob. Furthermore, the resulting state is pure conditioned on the classical communication. Therefore, we can apply  Lemma~\ref{lem:classical-attack} to derive an upper bound on the entropy of $X_0$ conditioned on Bob's system in the new protocol. Finally, we use the data-processing inequalities for the conditional entropy to show that this entropy can have decreased in the modified protocol by at most the number of commitments, which, together with inequality  \eqref{eq:alifan}, implies the statement of the theorem.

In the modified protocol, Alice, instead of sending bits to the commitment functionality, measures the bits to be committed, stores a copy of each and sends them to Bob, who stores them in a classical register, $C_A$. When one of these commitments is opened, he moves the corresponding bit to his register $B$. Bob simulates the action of the commitment functionality locally as follows: Instead of measuring a register, $Y$, and sending the outcome to the commitment functionality, he applies the isometry $U:\ket{y}_Y\mapsto \ket{yy}_{YY'}$ purifying the measurement of the committed bit and stores $Y'$ in another register, $C_B$. When Bob has to open the commitment, he measures $Y'$ and sends the outcome to Alice over the classical channel. The state of the modified protocol is pure conditioned on the classical communication.  Let $\rho_{ABC}^c$, where $C$ stands for $C_AC_B$, be the final state of this protocol. Note that its marginal state $\rho_{AB}^c$ is the corresponding state at the end of the original 
protocol. Since the protocol is $\eps$-secure for Bob, we have $\dis(\rho_A^0,\rho_A^1) \leq 2\eps$. From Lemma~\ref{lem:classical-attack} follows that there exists a unitary $U_{BC}$ such that Bob can transform the state $\rho^1$ into the state $\rhob^0$ with $\dis(\rho^0, \rhob^0) \leq 2 \sqrt{\eps}$. Since given the state $\rho_{X_0B}^0$, $X_0$ can be guessed from $\rho_B^0$ with probability $1-\eps$, it follows from~\eqref{eq:guess0} that $X_0$ can be guessed from $\rho_{BC}^1$ with a probability of at least $1- \eps - 2 \sqrt{\eps}$. By inequality~\eqref{eq:fano-q}, we obtain
\begin{align}\label{ineq:ent-low-bound}
\chvn{X_0}{B}{\rho^1} & \leq h(\eps) + h(2 \sqrt{\eps}) + (\eps + 2 \sqrt{\eps}) \cdot k\;.
\end{align}
We can use Lemma \ref{lem:data-processing-bound} and inequality~\eqref{chain-rule-c} to conclude that 
\begin{align*}
\chvn{X_0}{BC_AC_B}{\rho^1}\geq \chvn{X_0}{B}{\rho^1}-n\;.
\end{align*}
For $\eps \leq 0.002$, we have $h(\sqrt{\eps})>11h(\eps)$ and $21\eps < \sqrt{\eps}$. This implies the statement.
%%%\qed
\end{proof}
Theorem \ref{thm:imposs:com} implies that there exists a constant $c>0$ such that any protocol that implements $m+1$ bit commitments from $m$ bit commitments must have an error of at least $c/m$, i.e., bit commitment cannot be extended by quantum protocols. This result can be generalized in the following sense: For any protocol that implements a single string commitment from a certain number of bit commitments, the length of the implemented string commitments is essentially bounded by the number of used bit commitments, even if the protocol is allowed to have a small constant error \cite{WTHR11}. 

Next, we consider protocols where the two players have access to distributed randomness $P_{UV}$. We can model this primitive as a quantum primitive $\sum_{u,v} \sqrt{P_{UV}(u,v)} \cdot \ket{u,v}_{UV} \otimes \ket{u,v}_E$ that distributes the values $u$ and $v$ to Alice and Bob and keeps the register $E$. 
\begin{theorem}\label{thm:imposs:rand}
To implement a $\OT{1}{2}{k}$ with an error of at most $\eps$, where $0\leq \eps \leq 0.002$, from a primitive $P_{UV}$, we need 
%\[\max_v\log|\support( P_{U|V=v})|+\max_u\log|\support(P_{V|U=u})| \geq  (1- 3 \sqrt{\eps}) \cdot k - 3 h(\sqrt{\eps}),\]
\[\Chmaxeps{}{U}{V}+\Chmaxeps{}{V}{U}\geq  (1- 3 \sqrt{\eps}) \cdot k - 3 h(\sqrt{\eps})\;,\]
and
\[2\hvn{UV}{} \geq  (1- 3 \sqrt{\eps}) \cdot k - 3 h(\sqrt{\eps})\;.\]
\end{theorem}
\begin{proof}
 Let the final state of the protocol on Alice's and Bob's system be $\rho_{AB}^c$, when both players are honest and Bob has input $c \in \{0,1\}$. As in the proof of the previous theorem we have
\[\chvn{X_0}{B}{\rho^1} \geq (1- 20 \eps) \cdot k - 2h(5 \eps) \geq (1- 20 \eps) \cdot k - 10 h(\eps)\;.\]
Consider a modified protocol that starts from a state 
\begin{align*}
 \ket{\psi}_{UVU'V'}=\sum_{u,v} \sqrt{P_{UV}(u,v)} \cdot \ket{u,v}_{UV} \kron \ket{u,v}_{U'V'}\;,
\end{align*}
where the systems $V$ and $U'V'$ belong to Bob. Again Bob is more powerful in the modified protocol because he can simulate the state of the original protocol locally. As in \eqref{ineq:ent-low-bound} in the proof of the previous theorem we can, therefore, conclude that
\begin{align*}
\chvn{X_0}{BU'V'}{\rho^1} & \leq h(\eps) + h(2 \sqrt{\eps}) + (\eps + 2 \sqrt{\eps}) \cdot k\;.
\end{align*}
Since measuring register $V'$ and discarding register $U'$ results in the state $\rho^1_{X_0B}$, we can use Lemma \ref{lem:data-processing-bound} and inequality \eqref{chain-rule-c} to obtain
\begin{align*}
\chvn{X_0}{BU'V'}{\rho^1}\geq &\chvn{X_0}{B}{\rho^1}-\max_v\log|\support( P_{U|V=v})|\\
&-\max_u\log|\support(P_{V|U=u})|\; .
\end{align*}
This implies the first statement. The second statement follows from the inequality
\begin{align}
 \chvn{X_0}{BB'}{\rho^1}\geq \chvn{X_0}{B}{\rho^1}-\hvn{B'}{}\;,
\end{align}
which is implied by~\eqref{eq:cond-mi}.
\end{proof}
The theorem immediately implies the following corollary.
\begin{corollary} \label{cor:imposs4}
To implement a $\OT{1}{2}{k}$  with an error of at most $\eps$, where $0\leq \eps \leq 0.002$, from $n$ instances of $\OT{1}{2}{1}$ in either direction, we must have
\[2n \geq  (1- 3 \sqrt{\eps}) \cdot k - 3 h(\sqrt{\eps})\;.\]
\end{corollary}

Theorem~\ref{thm:imposs:rand} implies that $\OT{1}{2}{1}$ cannot be extended by quantum protocols in the following sense: Given a protocol that implements $m+1$ instances of $\OT{1}{2}{1}$ from $m$ instances of $\OT{1}{2}{1}$ with an error $\eps$, we can apply this protocol iteratively and implement $4m$ instances of $\OT{1}{2}{1}$ from $m$ instances of $\OT{1}{2}{1}$ with an error of $\eps':=3m\eps$, assuming that Bob follows the protocol. Thus, Corollary~\ref{cor:imposs4} implies that $12\sqrt{\eps'}+3h(\sqrt{\eps'})/m\geq 2$ if $\eps'\leq 0.002$. Thus, $\eps'\geq 0.002$ and $\eps \geq \frac{1}{1500m}$. 
Hence, any quantum protocol that implements  $m+1$ instances of $\OT{1}{2}{1}$ from $m$ instances of $\OT{1}{2}{1}$ must have an error of at least $\frac{1}{1500m}$.

The second bound of Theorem \ref{thm:imposs:rand} also holds for more general primitives that generate a pure state $\ket{\psi}_{ABE}$, distributes registers $A$ and  $B$ to Alice and Bob and keeps the purification in its register $E$.
\begin{theorem} \label{thm:QQ}
 To implement a $\OT{1}{2}{k}$ with an error of at most $\eps$, where $0\leq \eps \leq 0.002$, from a primitive $\ket{\psi}_{ABE}$, we need 
\[2\hvn{E}{\psi} \geq  (1- 3 \sqrt{\eps}) \cdot k - 3 h(\sqrt{\eps})\;.\]
\end{theorem}

The proof of Theorem~\ref{thm:QQ} follows exactly the same reasoning as Theorem~\ref{thm:imposs:rand} and is omitted.

Next, we give an additional lower bound for reductions of OT to commitments that shows that the \emph{number} of commitments (of arbitrary size) used in any $\eps$-secure protocol must be at least $\Omega( \log(1/\eps) )$. We model the commitments as before, i.e., the functionality applies the isometry $U:\ket{y}_Y\mapsto \ket{yy}_{YY'}$ and stores $YY'$ in separate registers $E_A$ and $E_B$ for Alice and Bob.
%, but store the commitments of Alice and Bob separately in $E_A$ and $E_B$. 
The proof idea is the following: We let the adversary guess a subset $\mT$ of commitments that he will be required to open during the protocol. He honestly executes all commitments in $\mT$, but cheats in all others. If the adversary guesses $\mT$ right, he is able to cheat in the same way as in any protocol that does not use any commitments. %The proof is given in Appendix \ref{app:thm:imposs1}.
\begin{theorem} \label{thm:imposs1}
Any quantum protocol that implements $\OT{1}{2}{k}$ using $\kappa$ commitments (of arbitrary length) must have an error of at least $2^{-\kappa} / 36$.
\end{theorem}
\begin{proof}
We define $\eps:=2^{-\kappa} / 36$. Let $\rho_{AB E_A E_B}^c$ be the final state of an $\eps$-secure protocol, when both players are honest and Bob has input $c \in \{0,1\}$. We distinguish two cases. In the first case, we assume that $\dis(\rho_{A E_A}^{0},\rho_{A E_A}^{1}) \geq \eps':=1/18\;$. We let Bob be honest and let Alice apply the following strategy: She chooses a random subset $\mT$ of $[k]$. She executes all commitments in $\mT$ honestly, but for all commitments not in $\mT$ she sends $\ket{0}$ to $E_A$ and simulates the action of the commitment functionality in her quantum register. Otherwise, she follows the whole protocol honestly.

During the execution of the protocol, Bob may ask Alice to open a certain set of commitments, $\mT'$. If $\mT' = \mT$, which happens with probability $2^{-\kappa}$ independently of everything else, then at the end of the protocol the global state is $\rho^c$, but $E_A$ is now part of Alice's system. Thus, the states of Alice's system for $c=0$ and $c=1$, have distance at least $\eps' \cdot 2^{-\kappa} > 2\eps$, which contradicts condition~\eqref{eq:secB}.

%Therefore, Alice has an advantage of more than $\eps'$ to distinguish $c=0$ from $c=1$ in this case. Thus, her advantage is more than $\eps' \cdot 2^{-\kappa} > 2\eps$, which contradicts condition~\eqref{eq:secB}.

In the second case, we assume that $\dis(\rho_{A E_A}^{0},\rho_{A E_A}^{1}) < \eps'$.
From condition~\eqref{eq:corr} follows that honest Bob can guess $X_1$ with probability $1- \eps$ if $c=1$. According to Lemma~\ref{lem:semihonest}, $X_0$ should be $5\eps$-close to uniform with respect to $\rho^1_{B}$. To obtain a contradiction to the security condition~\eqref{eq:secA2}, it is according to equation~\eqref{eq:guess} sufficient to show that Bob can guess the first bit of $X_0$ with a probability greater than $1/2 + 5\eps$.

% Let Alice be honest and Bob execute the same attack as Alice in the first case, choosing $c=1$. 
Again, if Bob guesses the set $\mT$ right, then $E_B$ is part of Bob's system. Then Lemma~\ref{lem:classical-attack} guarantees the existence of a unitary $U_{B E_B}$ such Bob can transform the state $\rho^1$ into a state $\rhob^1$ with $\dis(\rho^0,\rhob^1) \leq \sqrt{2\eps'}$. Thus, Bob can guess $X_0$ with an error of at most $ \sqrt{2\eps'} + \eps$ given $\rhob^1$. If he fails to guess $\mT$, he simply outputs a random bit as his guess for the first bit of $X_{0}$. Since the probability that he guesses the subset $\mT$ correctly is exactly $2^{-\kappa}$, he can guess the first bit of $X_{0}$ with probability
\begin{align*}
 (1 - 2^{-\kappa}) \cdot \frac12 &+ 2^{-\kappa} \cdot (1-\eps - \sqrt{2\eps'})\\
& = \frac12 + 2^{-\kappa} \cdot \left( \frac12 - \eps - \sqrt{2\eps'} \right) \\
& > \frac12 + 2^{-\kappa} \cdot \left( \frac12 - \eps'/2 - \sqrt{2\eps'} \right) \\
& = \frac12 + 2^{-\kappa} \cdot \frac{5}{36}\\
& = \frac12 + 5\eps\;.
\end{align*}
\end{proof}

\section{Reduction of OT to String Commitments}\label{sec:OTtoSC}
We will now show how to construct a protocol that is optimal with respect to the lower bounds of both Theorem~\ref{thm:imposs:com} and Theorem~\ref{thm:imposs1}. 
\begin{figure}[!t]

\begin{framed}
\noindent Protocol \textbf{OTfromCommitment \\}
\begin{enumerate}
\item Alice prepares $m$ EPR pairs, $(\ket{00}+\ket{11})/\sqrt{2}$, and sends one qubit of each pair to Bob. Bob selects $\hat \theta \in \{0, 1\}^m$ at random and measures the received qubits in basis $\hat \theta $, obtaining $\hat x \in \{0,1\}^m$. Alice chooses a basis $\theta \in \{0, 1\}^m$ at random (but does not measure her qubits yet).
\item Bob commits in blocks of size $b$ to $\hat \theta$ and $\hat x$. Alice samples a random subset $t \subseteq [\kappa]$ of cardinality $\alpha \kappa$ and asks Bob to open the commitments to the corresponding blocks of values  $(\hat \theta_i,\hat x_i)$. Let $\mT$ be the set of bits in $[m]$ corresponding to $t$. Alice measures her qubits indexed by $\mT$ in Bob's basis $\hat \theta_t$ to obtain $x_t$ and verifies that $x_i = \hat x_i$ whenever $\theta_i = \hat \theta_i$. If Bob does not commit to all values as required or does not open all commitments or if Alice detects an inconsistency, Alice outputs outputs two random $k$-bit strings $z_0,z_1$ and terminates the protocol.
\item (Set partitioning) Alice sends $\theta$ to Bob. Bob partitions $\bar \mT:=[m]\setminus \mT$ into the subsets $I_c = \{i \in \bar \mT :~\theta_i = \hat \theta_i \}$ and $I_{1-c}=\{i \in \bar \mT:~\theta_i\neq\hat \theta_i\}$ and sends $I_0$ and $I_1$ to Alice. %If Bob does not send the sets, Alice outputs two random $k$-bit strings $k_0,k_1$ and terminates the protocol. 
Additionally, Alice measures her qubits in basis $\theta$ to obtain $x$.
\item (Key extraction) Alice chooses and sends to Bob two-universal hash functions $f_0,f_1$ with output length $k$, and computes $z_0 := f_0(x_{I_0})$ and $z_1 := f_1(x_{I_1})$. Bob computes $z_c = f(\hat x_{I_0})$.
\end{enumerate}
\end{framed}
\end{figure}

We modify the protocol from \cite{BBCS92} by grouping the $m$ pairs of values into $\kappa$ blocks of size $b := m / \kappa$. We let Bob commit to the blocks of $b$ pairs of values at once. The subset $\mT$ is now of size $\alpha \kappa$, and defines the blocks to be opened by Bob. If Bob is able to open all commitments in $\mT$ correctly, then the state of the protocol must be close  in a certain sense to the state that would result from correctly measuring all qubits. Since we consider security in the malicious model, a dishonest player may abort the protocol by not sending any message. A possibility to handle this would be to include a special output \texttt{aborted} to the definition of the primitive. Here, we take the following, different approach, which is also used, for example, in~\cite{KoWeWu09}: Whenever a player does not send a (well-formed) message, the other player assumes that a fixed default message as, for example, the all-zero string has been sent. Note that our protocol is different from 
the protocols analyzed in~\cite{DFLSS09,BF09}. Besides replacing the bit commitments by strings commitments, Alice outputs two random strings if Bob aborts in the commitment or in the check step. This allows us to implement an ideal OT functionality that does not have a special output \texttt{aborted}. %Furthermore, the proof of security for Alice seems to get more involved without these modifications. 
%Note that we let Alice output two uniform random strings in case that her check on Bob's commitments fails, i.e., the protocol does not abort in this case.

We only need to estimate the error probability of the classical sampling strategy that corresponds to the new checking procedure of Alice and apply the result from \cite{BF09}. We need the following sampling result, which follows from the inequalities in \cite{Hoeffd63} (the proof is given in Appendix \ref{app:lemmas}).

\begin{lemma} \label{lem:sampling2}
Let $\alpha \in [0,\frac12]$. Let $y = (y_1, \dots y_m)$ be a bit string of length $m := b \kappa$ that we group into $\kappa$ blocks of size $b$. Let $\mT^*$ be a random subset of $[\kappa]$ of size $\alpha \kappa$, $\mT$ the corresponding set of bits in $[m]$ and $\bar \mT$ the complement of $\mT$. Let $\mT'$ be a random subset of $\mT$, where every element is chosen to be in $\mT'$ with probability $\frac12$, independently of everything else.  We have for any $\delta > 0$
\begin{align*}
\Pr \bigg [ \frac 1 {|\mT'|} \sum_{i \in \mT'} y_i \leq \frac 1 {(1 - \alpha) m} \sum_{i \in \bar \mT} y_i - \delta \bigg ] \leq \eps\;,
\end{align*}
where $\eps := 3 \exp(- (1/2 - \eps) \alpha \kappa \delta^2 /2)$.
\end{lemma}

\begin{lemma}[Security for Alice]\label{lem:sec:Alice}
Let $Z_0$ and $Z_1$ be the strings from $\{0,1\}^k$ output by Alice. Then there exists a binary $C$ such that for any $\eps,\delta > 0$ the following upper bound on the distance from uniform of $Z_{1-C}$ with respect to $Z_C$ and Bob's system holds:
\begin{align}\label{sec:Alice}
\dis(\rho_{Z_{1-C}Z_CEC},\tau_{\{0,1\}^{k}} \otimes &\rho_{Z_CEC})\nonumber\\
&\leq2^{-\frac 1 2 ((\frac 1 4 -\eps/2 -h(\delta))(1-\alpha)m-k)-1}\nonumber\\
&~+2\exp(-2\eps^2(1-\alpha)m)\nonumber\\
&~+\sqrt{3}\exp(-\alpha'\kappa\delta^2/4)\;,
\end{align}
where $E$ denotes the quantum state output by Bob, $\id$ the identity operator on $\mathbb{C}^{2^k}$ and $\alpha' := (1/2 - \delta) \alpha$. 
\end{lemma}
\begin{proof}
\cancel{
Let $K_0$ and $K_1$ be the strings from $\{0,1\}^k$ output by Alice. We show that there exists a bit $c$ such that $K_{1-c}$ is close to uniform with respect to Bob's system (given $K_{c}$), i.e., for any $\eps,\delta > 0$, we have, with $\alpha' := (1/2 - \delta) \alpha$,
\begin{align}\label{sec:Alice}
\dis(\rho_{K_{1-c}K_cE},\frac{1}{2^k}\id \otimes \rho_{K_cE})\nonumber\leq& \frac{1}{2}2^{-\frac 1 2 ((\frac 1 4 -\eps/2 -h(\delta))(1-\alpha)m-k)}\nonumber\\
&~+2\exp(-2\eps^2(1-\alpha)m)\nonumber\\&+\sqrt{3}\exp(-\alpha'\kappa\delta^2/4)
\end{align}
where $E$ denotes the quantum state output by Bob and $\id$ the identity operator on $\mathbb{C}^{2^k}$. 
}%Let $\ket{\varphi_{AE_o}}\in \mH_{A_1}\otimes\dots\mH_{A_m}\otimes \mH_{E_o},$ be 
We consider the state shared between Alice and Bob after Bob has committed to the bases $\hat\theta$ and the measurement outcomes $\hat x$ where we can assume $\hat\theta=\hat x=(0,\dots,0)$. Since we want to prove an upper bound on  \eqref{sec:Alice}, we can assume that Bob always opens all commitments. Otherwise the distance from uniform can only decrease. Alice now chooses a subset $\mT$ to be opened by Bob. As in the proof of Theorem 4 from \cite{BF10}, Lemma~\ref{lem:sampling2} implies that the joint state is $\sqrt{3}\exp(-\alpha'\kappa\delta^2/4)$-close to an ideal state that is for every choice of $\mT$ and $\mS$ in a superposition of states with relative Hamming weight in a $\delta$-neighbourhood of $\beta$ within $A_{\bar \mT}$, where $\beta$ is the ratio of inconsistencies that Alice detects and $\mS$ is the subset of $\mT$ that Alice checks. We assume that the state equals this ideal state and add the error later. Then, following the proof of Theorem 4 in \cite{BF09} for $\beta=0$, we obtain 
that the  distance from uniform of one of the outputs with respect to Bob's system (given the other output) is bounded from above by
\begin{align*}
 2^{-\frac 1 2 ((\frac 1 4 -\eps/2 -h(\delta))(1-\alpha)m-k)-1}+2\exp(-2\eps^2(1-\alpha)m)\;.
\end{align*}

If $\beta > 0$, the distance from uniform is zero. Thus, the statement follows by adding the distance of the ideal state to the real state.
\end{proof}
\begin{lemma}[Security for Bob]\label{lem:sec:Bob}
The protocol is perfectly secure for Bob.
\end{lemma}

\begin{proof}
Let $\rho_{A'YC}$ be the state created by the protocol if Bob is honest. We consider a hypothetical protocol where Bob does not use any commitments. He stores all the qubits received from Alice. After Alice sends the set $\mT$, he chooses a basis $\hat \theta$  and measures his qubits corresponding to $\mT$ to obtain $\hat x_{\mT}$ in basis $\theta$, but does not yet measure the other qubits. Then he sends $\hat x_\mT$ and $\hat \theta_\mT$ to Alice. After he gets the basis $\theta$ from Alice he measures all his remaining qubits in Alice's basis $\theta$ to obtain $\hat x_{\bar \mT}$. Next, he chooses his input $C \in \{0,1\}$ and constructs the sets $I_0$ and $I_1$ using $\theta$ and $\hat \theta$ as in the protocol. After receiving $f_{0}, f_{1} \in \mF$ from Alice, he computes $z_0 = f_{0}(\hat x_{I_{0}})$ and $z_1 = f_{1}(\hat x_{I_{1}})$ . This results in a state $\sigma_{A'Z_0Z_1C}$, where $Z_0$ and $Z_1$ are the values computed by Bob. We have $\sigma_{A'Z_0Z_1C}=\sigma_{A'Z_0Z_1}\otimes\sigma_{C}$ 
and $\sigma_{A'Z_CC}=\rho_{A'YC}$.
\end{proof}

\begin{lemma}[Correctness]\label{lem:sec:cor}
The protocol is perfectly correct.% an error of at most $m\cdot2^{-\tfrac{b}{8}+1}$.
\end{lemma}
\begin{proof}
If both players are honest, then $Z_0$, $Z_1$ and $C$ are independently distributed according to the required distributions. Furthermore, Bob always computes $Z_C$ as his output.
\end{proof}

The following theorem is then immediately implied by Lemmas~\ref{lem:sec:Alice},~\ref{lem:sec:Bob} and~\ref{lem:sec:cor}.
\begin{theorem}\label{thm:optimal-protocol}
There exists a quantum protocol that uses $\kappa = O(\log 1/\eps)$ commitments of size $b$, where $\kappa b = O(k + \log 1 / \eps)$, and implements a $\OT{1}{2}{k}$ with an error of at most~$\eps$.% according to the security definition in~\cite{FS09}.
\end{theorem}
%\severin{security for Bob and correctness}

%% file: malicious_peerreview.tex
\section{Malicious OT implies Semi-honest OT}\label{appendix:sec:mal}

In the malicious model the adversary is not required to follow the protocol. Therefore, a protocol that is secure in the malicious model protects against a much bigger set of adversaries. On the other hand, the security definition in the malicious model only implies that for any (also semi-honest) adversary there exists a \emph{malicious} simulator for the ideal primitive, i.e., the simulator is allowed to change his input or output from the ideal primitive. Since this is not allowed in the semi-honest model, security in the malicious model does not imply security in the semi-honest model in general. For implementations of OT\footnote{And any other so-called deviation revealing functionality.}, however, it has been shown in  \cite{PR08} that this implication \emph{does} hold, because if the adversary is semi-honest, a simulator can only change the input with small probability. Otherwise, he is not able to correctly simulate the input or the output of the protocol. Therefore, any impossibility result for OT in the semi-honest model also implies impossibility in the malicious model.

%We will state these results for $\OT{1}{n}{k}$ and $\RabinOT{p}{k}$ with explicit bounds on the errors.
We will state these result for $\OT{1}{n}{k}$ with explicit bounds on the errors.
\begin{lemmaC}
\label{lem:malicious:ot}
 If a protocol implementing $\OT{1}{n}{k}$ is secure in the malicious model with an error of at most $\eps$, then it is also secure in the semi-honest model with an error of at most $(2n+1)\eps$.
\end{lemmaC}

\begin{proof}
From the security of the protocol we know that there exists a (malicious) simulator that simulates the view of honest Alice. If two honest players execute the protocol on input $(x_0,\dots,x_{n-1})$ and $c$, then with probability $1-\eps$ the receiver gets $y = x_c$. Thus, the simulator can change the input $x_i$ with probability at most $2\eps$ for all $0 \leq i < n-1$. We construct a new simulator that executes the malicious simulator but never changes the input. This simulation is $(2n+1)\eps$-close to the distribution of the protocol. From the security of the protocol we also know that there exists a (malicious) simulator that simulates the view of honest Bob. If two honest players execute the protocol with uniform input $(X_0,\dots,X_{n-1})$ and choice bit $c$, then with probability $1-\eps$ the receiver gets $y = x_c$. If the simulator changes the choice bit $c$, he does not learn $x_c$ and the simulated $y$ is not equal to $x_c$ with probability at least $1/2$. Therefore, the simulator can change $c$ or the output with probability at most $4\eps$. As above we can construct a simulator for the semi-honest model with an error of at most $5\eps$. 
%\qed
\end{proof}
\cancel{
\begin{lemmaC}
\label{lem:malicious:rabin}
If a protocol implementing $\RabinOT{p}{k}$ is secure in the malicious model with an error of at most $\eps$, then it is also secure in the semi-honest model with an error of at most $\max(\frac{2^{k+1}}{2^k-1}\eps + 2\eps,2\eps/p)$.
\end{lemmaC}
\begin{proof}
From the security of the protocol we know that there exists a (malicious) simulator that simulates the view of honest Alice. If two honest players execute the protocol on input $x$, then with probability at most $\eps$ the receiver gets an output $x'\notin \{x,\Delta\}.$ Thus, the simulator can change the input $x$ with probability at most $2\eps/p$.  From the security of the protocol we also know that there exists a (malicious) simulator that simulates the view of honest Bob. Let the input be chosen uniformly. If the simulator changes the output from  $\Delta$ to $y'$, then with probability at most $1/2^k$ it holds that $y'=x$. Thus, the simulator may change the output with probability at most $\frac{2^{k+1}}{2^k-1}\eps/(1-p)$ from $\Delta$. Therefore the simulator may change an output $x\neq \Delta$ with probability at most $\frac{2^{k+1}}{2^k-1}\eps/(1-p) + 2\eps$. Otherwise the probability that $x'\notin \{x,\Delta\}$ is greater than $2\eps$. As in lemma \ref{lem:malicious:ot} we can now construct semi-honest simulators with an error of at most $\max(\frac{2^{k+1}}{2^k-1}\eps/(1-p) + 2\eps,2\eps/p)$.
\qed
\end{proof}
}
Note that some of our proofs could easily be adapted to the malicious model to get slightly better bounds than the ones that follow from the combination of the bounds in the semi-honest model and Lemma \ref{lem:malicious:ot}.% and \ref{lem:malicious:rabin}. 

%% file: smooth_entropies_ieee_peerreview.tex
\section{Smooth Entropies}\label{app:smooth}
In the following we prove different properties of the entropies $\Chminteps{\eps}{X}{Y}$ and $\Chmaxeps{\eps}{X}{Y}$. Note that some of these properties (or special cases of them) have already been shown in~\cite{RenWol05}.% or stated without a proof there. 

We first introduce the following auxiliary quantities. 

\begin{definition}
For random variables $X,Y$ and $\eps \in [0,1)$, we define 
\begin{align*}
 \rmaxE{\eps}{X}{Y}&:=\min_{\Omega:\Pr[\Omega]\geq 1-\eps}\max_{y \in \mY}|\supp{P_{X\Omega|Y=y}}|\text{ and}\\
 \rminTE{\eps}{X}{Y}&:=\min_{\Omega:\Pr[\Omega]\geq 1-\eps}\sum_{y \in \mY} P_Y(y)\max_{x}P_{X\Omega|Y=y}(x)\;.\end{align*}
%where $\mY$ is the support of $P_Y$.
\end{definition}
 Note that $\Chminteps{\eps}{X}{Y}=-\log \rminE{\eps}{X}{Y}$ and $\Chmaxeps{\eps}{X}{Y}=\log \rmaxE{\eps}{X}{Y}$.

The following lemma shows that the smooth conditional max-entropy is subadditive.
\begin{lemma}[Subadditivity]\label{lem:sub-add}
Let $X,Y,Z$ be random variables and $\eps,\eps' \geq 0$ such that $\eps + \eps' \in [0,1)$. Then
 \begin{align*}
    \Chmaxeps{\eps+\eps'}{XY}{Z}&\leq \Chmaxeps{\eps}{X}{Z}+\Chmaxeps{\eps'}{Y}{XZ}\;.
 \end{align*}
\end{lemma}
\begin{newproof}
Let $\Omega$ be an event with $\Pr[\Omega]\geq 1-\eps$ and 
\[\max_{x,z} |\supp{P_{Y\Omega|X=x,Z=z}}|\leq \rmaxE{\eps}{Y}{XZ}\;.\]
Let $\Omega'$ be an event with $\Pr[\Omega']\geq 1-\eps'$ and $\Markov{\Omega'}{(X,Z)}{(Y,\Omega)}$ such that 
\[\max_z |\supp{P_{X\Omega'|Z=z}}|\leq \rmax^\eps(X|Z)\;.\]
Then $\Pr[\Omega,\Omega']\geq 1-\eps-\eps'$ and 
\[\rmax^{\eps+\eps'}(XY|Z)\leq \max_z |\supp{P_{XY\Omega\Omega'|Z=z}}|\;.\]
We have
 \begin{align*}
  &\max_z |\supp{P_{XY\Omega\Omega'|Z=z}}|\\
	&~~\leq\max_z(|\supp{P_{X\Omega'|Z=z}}|\cdot \max_x |\supp{P_{Y\Omega|X=x,Z=z}}|)\\
  &~~\leq \max_z|\supp{P_{X\Omega'|Z=z}}|\cdot \max_{x,z} |\supp{P_{Y\Omega|X=x,Z=z}}|\;.\\&& \IEEEQED
   %&\leq \max_z|\supp(P_{Y\Omega'|Z=z})|\cdot \max_{z} |\supp(P_{Y\Omega|Z=z})|
 \end{align*}
%\qedhere
\end{newproof}

Next, we show that conditioning on an additional random variable cannot reduce the conditional smooth entropies. 
\begin{lemma}\label{lem:mono2}
Let $X,Y,Z$ be random variables and $\eps \in [0,1)$. Then
\begin{align*}
  \Chminteps{\eps}{X}{Z}&\geq \Chminteps{\eps}{X}{YZ}\;.
  %\Chmaxeps{\eps}{XY}{Z}&\geq \Chmaxeps{\eps}{X}{YZ}.
\end{align*} 
\end{lemma}
\begin{proof}
 %The first inequality follows immediately from Lemma \ref{lem:chain-rule}.
 Let $\Omega$ be an event with $\Pr[\Omega]\geq 1-\eps$. Then 
\begin{align*}
\sum_{z}P_Z(z) &\max_{x} P_{X\Omega \mid Z=z}(x)\\
 &= \sum_{z}P_Z(z)  \max_{x} \sum_{y} P_{Y|Z=z}(y) P_{X\Omega \mid Y=y,Z=z}(x) \\
%&\leq  \max_{x,z} \sum_{y} P_{Y|Z=z}(y) \max_{x,y,z} P_{X \mid Y=y,Z=z}(x)\\
&\leq \sum_{z}P_Z(z) \max_{x,y} P_{X\Omega \mid Y=y,Z=z}(x)\;.%&&&\mbox{\qedhere}
\end{align*}
\end{proof}

The Shannon entropy satisfies the inequality $\cHS{X}{Z}-\cHS{X}{YZ}=\cIS{X}{Y}{Z}\leq \cHS{Y}{Z}$. The next lemma shows that this property can be generalized to the smooth min- and max-entropy.

\begin{lemma}\label{lem:mutual-inf2}
Let $X,Y,Z$ be random variables and $\eps,\eps' \geq 0$ such that $\eps + \eps' \in [0,1)$. Then
 \begin{align*}
    \Chminteps{\eps}{X}{Z}-\Chmaxeps{\eps'}{X}{YZ}\leq \Chminteps{\eps+\eps'}{Y}{Z}\;.
 \end{align*}
\end{lemma}
\begin{proof}
Let $\Omega$ be an event with $\Pr[\Omega]\geq 1-\eps$ and 
\[\sum_z\max_{x}P_{XZ\Omega}(x,z)\leq \rminTE{\eps}{X}{Z}\;.\]
Let $\Omega'$ be an event with $\Pr[\Omega']\geq 1-\eps'$ such that
\[\max_{y,z}|\supp{P_{X\Omega'|Y=y,Z=z}}|\leq \rmax^\eps(X|YZ)\;.\]
Then $\Pr[\Omega,\Omega']\geq 1-\eps-\eps'$ and 
\[\rminTE{\eps+\eps'}{Y}{Z}\leq \sum_{z}P_Z(z)\max_{y} P_{Y\Omega\Omega'|Z=z}(y)\;.\]
We have for all $z$
\begin{align*}
 \max_{x,y} P_{XY\Omega\Omega'|Z=z}(x,y)&\leq \max_{x,y} P_{XY\Omega|Z=z}(x,y)\\
&\leq \max_{x} P_{X\Omega|Z=z}(x)\;.
\end{align*}
Furthermore, we have
\begin{align*}
 |\{x: P_{XY\Omega\Omega'|Z=z}(x,y) > 0\}|\leq  |\supp{P_{X\Omega'|Y=y,Z=z}}|\;.
\end{align*}
Together, we obtain
\begin{align*}
     \rminTE{\eps+\eps'}{Y}{Z} &\leq \sum_{z}P_Z(z)\max_{y} P_{Y\Omega\Omega'|Z=z}(y)\\
     &=\sum_{z}P_Z(z)(\max_{y}\sum_{x} P_{XY\Omega\Omega'|Z=z}(x,y))\\
     %&\leq\sum_{z}P_Z(z)(\max_{y,z} |\supp{P_{XY\Omega\Omega'|Z=z}}|\cdot \max_{x,y} P_{XY\Omega\Omega'|Z}(x,y|z)\\ 
     &\leq\sum_{z}P_Z(z)(\max_{y,z} |\supp{P_{X\Omega\Omega'|Y=y,Z=z}}|\\
		 &~~~~\cdot \max_{x,y} P_{XY\Omega\Omega'|Z=z}(x,y))\\
     &\leq\max_{y,z} |\supp{P_{X\Omega'|Y=y,Z=z}}|\\
		 &~~~~\cdot\sum_{z}P_Z(z) \max_{x} P_{X\Omega|Z=z}(x)\\
     &\leq \rminTE{\eps}{X}{Z} \cdot \rmax^{\eps'}(X|YZ)\;.
     %&\leq\max_{y,z} |\supp(P_{X\Omega'|Y=y,Z=z})|\cdot \sum_{z}P_Z(z)\max_{x} P_{X\Omega|Z=z}(x)\\
\end{align*}

%\qed
\end{proof}

Note that the proof also implies the stronger inequality $\Hmin^{\eps}(XY|Z)-\Hmax^{\eps'}(X|YZ)\leq \Hmin^{\eps+\eps'}(Y|Z)$, which corresponds in a certain sense to the inequality $\cHS{X}{Z}-\cHS{X}{YZ}\leq \cHS{XY}{Z}$ for the Shannon entropy.

The following lemma shows that the smooth min-entropy $\Chminteps{\eps}{X}{Y}$ satisfies a data processing inequality, i.e., it cannot be decreased by additionally processing $Y$.

\begin{lemma}[Data Processing]\label{lem:markov-t}
Let $X,Y,Z$ be random variables with $\Markov{X}{Y}{Z}$ and $\eps \in [0,1)$. Then
\begin{align*}
  \Chminteps{\eps}{X}{Y} &\leq \Chminteps{\eps}{X}{YZ}\;.
\end{align*}
\end{lemma}
\begin{proof}
Let $\Omega$ be an event with $\Pr[\Omega]\geq 1-\eps$ and $\Markov{\Omega}{XY}{Z}$ such that
\begin{align*}
\rminTE{\eps}{X}{Y}&=\sum_y P_{Y}(y) \max_{x}P_{X\Omega|Y=y}(x)\;. 
\end{align*}

We have 
\begin{align*}
 P_{X\Omega \mid Y=y, Z=z}(x)&=P_{X \mid Y=y,Z=z}(x) P_{\Omega \mid X=x,Y=y,Z=z}\\
 &=P_{X \mid Y=y}(x) P_{\Omega \mid X=x,Y=y}\\
 &=P_{X\Omega|Y=y}(x)\;.
\end{align*}
Thus, we obtain
\begin{align*}
 \rminTE{\eps}{X}{YZ}&\leq \sum_{y,z}P_{YZ}(y,z)\max_{x} P_{X\Omega \mid Y=y, Z=z}(x)\\
 &=\sum_y P_{Y}(y) \max_{x}P_{X\Omega|Y=y}(x)\;.%&&\mbox{\qedhere}
\end{align*}

\end{proof}

The smooth max-entropy $\Chmaxeps{\eps}{X}{Y}$ also satisfies a data processing inequality, i.e., it cannot be decreased by additionally processing $Y$.
\begin{lemma}\label{lem:markov-max}
Let $X,Y,Z$ be random variables with $\Markov{X}{Y}{Z}$ and $\eps\in [0,1)$. Then
\begin{align*}
  \Chmaxeps{\eps}{X}{Y}&\leq\Chmaxeps{\eps}{X}{YZ}\;.
\end{align*}
\end{lemma}
\begin{proof}
Let $\Omega$ be an event such that 
\[\rmaxE{\eps}{X}{YZ}=\max_{y,z}|\supp{P_{X\Omega|Y=y,Z=z}}|\;.\]
For all $y$, we define $\eps_y:=P_{\Omega|Y=y}$. Let $\Omega_y$ be an event such that \[\rmaxE{\eps_y}{X}{Z,Y=y}=\max_{z}|\supp{P_{X\Omega_y|Y=y,Z=z}}|\;.\]

Let $\bar z_y$ be such that $P_{\Omega_y|Y=y,Z=\bar z}$ is maximal. We define $\bar \Omega_y$ with $P_{\bar \Omega_y|X=x,Y=y}:=P_{\Omega_y|X=x,Y=y,Z=\bar z}$. Then, we have $P_{\bar \Omega_y|Y=y}\geq P_{\Omega_y|Y=y}\geq 1-\eps_y$ and $P_{X\Omega_y|Y=y,Z=z}\geq P_{X\Omega_y|Y=y,Z=\bar z}= P_{X\bar \Omega_y|Y=y}$ and, therefore, \[\rmaxE{\eps_y}{X}{Z,Y=y}\geq \rmaxE{\eps_y}{X}{Y=y}\;.\] Thus, we get
\begin{align*}
\rmaxE{\eps}{X}{YZ}&=\max_{y,z}|\supp{P_{X\Omega|Y=y,Z=z}}|\\
&\geq \max_y \rmaxE{\eps_y}{X}{Z,Y=y}\\
&\geq \max_y \rmaxE{\eps_y}{X}{Y=y}\\
&\geq \rmaxE{\eps}{X}{Y}\;.%&\mbox{\qedhere}\\
\end{align*}
\end{proof}

The smooth max-entropy satisfies the following monotonicity properties.
\begin{lemma}\label{lem:mono-max}
Let $X,Y,Z$ be random variables and $\eps \in [0,1)$. Then
\begin{align*}
  %\Chminteps{\eps}{X}{Z}&\geq \Chminteps{\eps}{X}{YZ}.
  \Chmaxeps{\eps}{XY}{Z}\geq \Chmaxeps{\eps}{X}{Z}\geq\Chmaxeps{\eps}{X}{YZ}\;.
\end{align*} 
\end{lemma}
\begin{proof}
 %The first inequality follows immediately from Lemma \ref{lem:chain-rule}.
 Let $\Omega$ be an event with $\Pr[\Omega]\geq 1-\eps$. Then the first inequality follows from
\begin{align*}
 \max_{z} | \supp{P_{XY\Omega \mid Z=z}}|\geq \max_{z}|\supp{P_{X\Omega \mid Z=z}}|\;.
\end{align*}
and the second inequality from
\begin{align*}
 \max_{y,z} | \supp{P_{X\Omega \mid Y=y,Z=z}}|\leq \max_{z}|\supp{P_{X\Omega \mid Z=z}}|\;.%&\mbox{\qedhere}
\end{align*}
\end{proof}

%% file: Lemmas_peerreview.tex
\section{Technical Lemmas}\label{app:lemmas}
%Lemma \ref{lem:semihonest} shows that if Bob knows $X_1$ with a small error, then the security condition implies that $X_0$ is close to uniform with respect to his state, if Alice chooses her inputs at random.

\begin{lemma2}{\ref{lem:h-smooth}}
Let $(X,Y)$ and $(\hat X,\hat Y)$ be random variables distributed according to $P_{XY}$ and $P_{\hat X \hat Y}$, and let $\dis(P_{XY},P_{\hat X \hat Y})\leq \epsilon$. Then 
\begin{align*}
\cHS{\hat X}{\hat Y}&\geq\cHS{X}{Y}-\epsilon \log |\mathcal{X}|-h(\epsilon)\;.
%\Ishannon(\hat X;\hat Y)&\geq \Ishannon(X;Y)-\epsilon \log(|\mathcal{X}|)-\hop(\epsilon)\;.
\end{align*}
\end{lemma2}
\begin{proof}
There exist random variables $A,B$ such that $P_{XY|A=0}=P_{\hat X\hat Y|B=0}$ and $\Pr[A=0]=\Pr[B=0]=1-\epsilon$. Thus, using the monotonicity of the entropy and the fact that $\cHS{X}{YB=1}\leq \log |\mX|$ we get that
\begin{align*}
\cHS{\hat X}{\hat Y} &\geq (1-\eps)\cHS{\hat X}{\hat YA=0}+\eps \cHS{\hat X}{\hat YA=1}\\
&\geq(1-\epsilon)\cHS{X}{YB=0}\\
&=\cHS{X}{YB}-\epsilon\cHS{X}{YB=1}\\
&=\cHS{XB}{Y}-\cHS{B}{Y}-\epsilon\cHS{X}{YB=1}\\
&\geq\cHS{X}{Y}-h(\epsilon)-\epsilon \log |\mathcal{X}|\;.
\end{align*}
%%\qed
\end{proof}

\begin{lemma} \label{lem:semihonest}
Let $\rho_{X_0 X_1 B}$ satisfy conditions~\eqref{eq:secA1} and~\eqref{eq:secA2}.
If there exists a measurement $G$ on system $B$ such that $\Pr[ G(\rho_B) = X_1] \geq 1- \eps$, then
\begin{equation*}
\dis(\rho_{X_{0} X_{1} B}, \tau_{X_0} \otimes \rho_{X_{1} B}) \leq 5\eps\;.
\end{equation*}
\end{lemma}
\begin{proof}
Let $\sigma_{X_0 X_1 B C'}$ be the state in conditions~\eqref{eq:secA1} and~\eqref{eq:secA2}. Then \eqref{eq:guess} implies 
\[ \Pr [ G(\sigma_B) = X_1] \geq \Pr [ G(\rho_B) = X_1] - \eps \geq 1 - 2 \eps\;.\]
In the state $\sigma_{X_0 X_1 B C'}$, we can guess the first bit of $X_{1-C'}$ if we output the first bit of $G(\sigma^B)$ whenever $C'=0$ and a random bit otherwise. We succeed with a probability of
\begin{align*}
      g \geq  & \frac12 \cdot \Pr[C' = 1] + \Pr [ G(\sigma^B) = X_1 \wedge C' = 0] \\
  = & \frac12 \cdot (1- \Pr[C' = 0]) + \Pr[C' = 0] - \Pr [ G(\sigma^B) \neq X_1 \wedge C' = 0] \\
  \geq & \frac12 \cdot (1- \Pr[C' = 0]) + \Pr[C' = 0] - 2 \eps \\
  = & \frac12 + \frac{\Pr[C'=0]}{2} - 2 \eps\;.
\end{align*}
Since $X_{1-C'}$ is uniform with respect to the rest, we have $g \leq \frac12$ and, therefore,
$\Pr[C'=0] \leq 4\eps$. This implies that for
$\hat \sigma_{X_0 X_1 B C'} := \tau_{X_{0}} \otimes \sigma_{X_{1} B} \otimes \ketbra{1}{1}$ we have
\begin{equation*}
\dis(\sigma_{X_{1-C'} X_{C'} B C'}, \hat \sigma_{X_{1-C'} X_{C'} B C'}) \leq 4\eps
\end{equation*}
and hence
\begin{align*}
\dis(\rho_{X_{0} X_{1} B}, \tau_{X_0} \otimes \rho_{X_{1} B})&\leq \dis(\rho_{X_{0} X_{1} B}, \sigma_{X_{0} X_{1} B})\\
&\;+\dis(\sigma_{X_{0} X_{1} B},\hat \sigma_{X_{0} X_{1} B})\\ 
&\leq 5\eps\;.
\end{align*}
\end{proof}

We will use the following lemma from \cite{Hoeffd63} to prove Lemma \ref{lem:sampling2}.
\begin{lemma}\label{lem:hoeffding}
Let $\beta=(\beta_1,\cdots,\beta_n) \in [0,1]^n$. Let $X_1,\cdots,X_k$ denote a random sample with replacement and let $Y_1,\cdots,Y_k$ denote a random sample without replacement from $\beta$. Then the following two inequalities hold
\begin{align}
\Pr\left[\frac{1}{k}\sum_{i=1}^k X_i \leq \frac{1}{n}\sum_{i=1}^n \beta_i-\eps\right]\leq e^{-2k\eps^2}\;,\nonumber\\
\Pr\left[\frac{1}{k}\sum_{i=1}^k Y_i\leq \frac{1}{n} \sum_{i=1}^n \beta_i-\eps\right]\leq e^{-2k\eps^2}\;.\nonumber
\end{align}
\end{lemma}

\begin{proof}[Proof of Lemma~\ref{lem:sampling2}]
Let $a_j$ be the number of bits where $y$ is equal to $1$ in the $j$th block, for $j \in [\kappa]$, and let $\bar \mT^*$ be the complement of $\mT^*$.
We set $\beta_j := 1 - a_j/b$. Then all $\beta_j$ are in $[0,1]$. Thus, we can apply Lemma~\ref{lem:hoeffding} to obtain
\begin{align} \label{eq:1}
\Pr &\left [ \frac 1 {(1 - \alpha) m} \sum_{i \in \bar \mT} y_i  \geq \frac 1 {m} \sum_{i=1}^{m} y_i + \eps \right ]\nonumber\\
&~~~~~~~~~~~~ = \Pr \left [\frac 1 {(1 - \alpha) m} \sum_{j \in \bar \mT^*} a_j \geq \frac 1 {m} \sum_{j=1}^{\kappa} a_j + \eps \right ] \nonumber\\
&~~~~~~~~~~~~ \leq e^{-2(1 - \alpha) \kappa \eps^2}\;.
\end{align}
Let $S \in \{0, \dots, \alpha m\}$ be the size of $\mT'$. Even if we condition on the event that $\mT'$ has size $s$, i.e, $S=s$, $\mT'$ is still a random subset of $[m]$. Hence, we can apply Lemma~\ref{lem:hoeffding} again and obtain
\[\Pr \left [ \frac 1 {s} \sum_{i \in \mT'} y_i \leq \frac 1 {m} \sum_{i=1}^{m} y_i - \eps \Big | S = s \right ] 
\leq e^{-2s \eps^2}\;,
\]
which implies that
\[\Pr \left [ \frac 1 {S} \sum_{i \in \mT'} y_i \leq \frac 1 m \sum_{i=1}^m y_i - \eps \Big | S \geq \alpha' m \right ] 
\leq e^{- 2\alpha' m \eps^2}\;,
\] where $\alpha':=(1/2-\eps)\alpha$.
From Lemma~\ref{lem:hoeffding} follows that 
\[ \Pr[ S \leq \alpha' m] = \Pr \left [ \frac{S}{\alpha m}\leq \frac12 - \eps \right ] \leq e^{-2 \alpha m \eps^2}\;.\]
Hence,
\begin{align} \label{eq:2}
\Pr \left [ \frac 1 {S} \sum_{i \in \mT'} y_i \leq \frac 1 m \sum_{i=1}^m y_i - \eps \right ] 
&\leq e^{-2 \alpha m \eps^2} + e^{- \alpha' \kappa \eps^2 / 2} \nonumber\\
&\leq 2 e^{- \alpha' m \eps^2 / 2}\;.
\end{align}
Combining equations~(\ref{eq:1}) and (\ref{eq:2}), we obtain
\begin{align}
\Pr \left[ \frac 1 {S} \sum_{i \in \mT'} y_i \leq \frac 1 {(1 - \alpha) m} \sum_{i \in \bar \mT} y_i - 2\eps \right] &\leq 2 e^{- 2\alpha' m \eps^2}\nonumber\\&~~~~ + e^{-2(1 - \alpha) \kappa \eps^2}\label{eq:lem:3} \nonumber\\ 
& \leq 3 e^{- 2\alpha' \kappa \eps^2 } \nonumber\;. 
\end{align}
\end{proof}